\documentclass{article}

\usepackage{MacrosRidges}
\usepackage[ruled, vlined]{algorithm2e}
\usepackage{float}

\usepackage{xcolor}
\usepackage{cite}

\usepackage{fullpage}
\usepackage{mathtools}
\usepackage{graphicx}
\usepackage{subcaption}

\def\blambda{{\bs{\lambda}}}
\def\bLambda{{\bs{\Lambda}}}
\def\Vo{\V_{\!\!\perp}}
\def\Vp{\V_\parallel}
\def\Sigmahx{{\bs{\Sigma}_h(\x)}}
\def\muhx{{\bs{\mu}_h(\x)}}
\def\balpha{{\bs{\alpha}}}
\def\bgamma{{\bs{\gamma}}}

\begin{document}
\title{Log-Concave Ridge Estimation\footnote{This work was supported by Swiss National Science Foundation. It is part of the author's PhD dissertation.}}
\author{Christof Str\"ahl\\University of Bern, Switzerland}
\date{\today}
\maketitle

\begin{abstract}
  We develop a density ridge search algorithm based on a novel density
  ridge definition. This definition is based on a conditional variance
  matrix and the mode in the lower dimensional subspace.  It is compared to
  the subspace constraint mean shift algorithm in \citet{OzertemETAL2011},
  based on the gradient and Hessian of the underlying probability density
  function. We show the advantages of the new algorithm in a simulation
  study and estimate galaxy filaments from a data set of the Baryon
  Oscillation Spectroscopic Survey.
\end{abstract}

\tableofcontents

\section{Introduction}
\label{sec:introduction}

\subsection{From Principle Curves to Ridges}
\label{sec:from-princ-curv}

Nowadays, a ridge of a probability distribution is understood as a lower
dimensional structure, where each point on the ridge is the mode in an affine
subspace. The subspace is given by some eigenvectors of the Hessian of the
underlying probability density function; see \citet{OzertemETAL2011},
\citet{Genovese2014}. We give a more general definition, where the subspace
is given by the directions with smallest local variance.

An early approach to define lower dimensional structures of a distribution
were made by \citet{HastieStuetzle1989}. In their paper, a principle curve
is a smooth curve through the middle of the data. For any point on the
curve the average of all data points coincides with the point. More
formally, they give a definition for principle curves of probability
density functions as follows:

\begin{defi}
  \label{def:Principal-Curves-Hastie}
  Let $f$ be a probability density function on $\R^d$ with corresponding
  random vector $\X$ with finite second moments and assume without
  loss of generality that $\Ex(\X) = 0$.  Let $\gamma$ be a differentiable
  unit-speed curve in $\R^d$ parameterized over $I \subset \R$,
  i.e.~$\lVert \gamma'(t)\rVert = 1$ for $t \in I$, that does not
  intersect itself and has finite length inside any finite ball in
  $\R^d$. Define the \emph{projection index} $t_{\gamma}: \R^d \rr \R$ as
  \begin{displaymath}
    t_\gamma(\x) = \sup_t \bigl\{ t : \lVert \x - \gamma(t) \rVert =
    \inf_s \lVert \x - \gamma(s) \rVert \bigr\} .
  \end{displaymath}
  The curve $\gamma$ is called \emph{self-consistent} or a
  \emph{principle curve} of $f$ if $\Ex(\X \, \vert \, t_\gamma(\X) =
  t) = \gamma(t)$ for a.e.~$t$.
\end{defi}

The intuition behind the principle curves in Definition
\ref{def:Principal-Curves-Hastie} is that for any parameter value $t$ we
collect the points projected on $\gamma(t)$ on the principle curve, and
their average should lie on the principle curve. A distribution may have
multiple principle curves, i.e.~for any spherically symmetric distribution
any straight line trough the center is a principle curve. The existence of
principal curves remains an open question, expect for very special cases, i.e~
elliptical distributions.

\citet{Kegl2000} address the issue in \citet{HastieStuetzle1989}, that
principle curves do not exist for any distribution. To resolve this
problem they generalize a property of principle components: A straight line
$\gamma(t)$ is the first principle component, if and only if,
\begin{displaymath}
  \Ex\bigl( \min_t\lVert \X - \gamma(t) \rVert^2 \bigr) \leq \Ex\bigl(\min_t
  \lVert\X - \hat{\gamma}(t) \rVert^2 \bigr)
\end{displaymath}
for any other straight line $\hat{\gamma}$. Instead of considering straight
lines only, they restrict the class of curves on those with finite length. The
finite length constraint is necessary, because otherwise the expected
squared distance between $\X$ and the curve becomes arbitrary small and the
length of the curve tends to infinity. The formal definition of a principle
curve is the following:

\begin{defi}
  \label{def:Principal-Curves-Kegl}
  A curve $\gamma$ is called \emph{principal curve of length $L$ for $\X$}
  if $\gamma$ minimizes
  \begin{displaymath}
    \Delta(\gamma) \defeq \Ex\bigl( \inf_t \lVert \X - \gamma(t) \rVert^2
    \bigr) = \Ex\bigl( \lVert \X - \gamma(t_\gamma(\X)) \rVert^2 \bigr)
  \end{displaymath}
over all curves of length less then or equal to $L$.
\end{defi}

Whenever $\X$ has finite second moments, a principle curve as in Definition
\ref{def:Principal-Curves-Kegl} exists. They present the polygonal line
algorithm to estimate the principle curve.

In \citet{Delicado2001} and \citet{Delicado2003}, a principle curve (of
oriented points) is defined as a curve contained in the set of oriented
points. A point $\x$ is oriented, whenever it holds
$\x = \Ex\bigl( \X \, \vert \, \X \in \text{H}(\x, \b) \bigr)$, where $\b$
is the unit vector, such that the total variance
$\tr \bigl(\Var( \X \, \vert \, \X \in \text{H}(\x, \b)) \bigr)$ is minimal for
the hyperplane
$\text{H}(\x, \b) \defeq \{\y \in \R^d : (\y - \x)^\top \b = 0\}$. That
means, the principle curve consists of averages of $\X$, given $\X$ lies in a
hyperplane orthogonal to the direction with largest variance.

All the concepts so far only consider one-dimensional
structures. Furthermore, they are based on the expectation and not on the
shape of the underlying density function.

This changes with \citet{OzertemETAL2011}. Their definition of a principle
curve is based on the gradient and Hessian matrix of the underlying
probability density function. They also generalize the definition from
one-dimensional curves to arbitrary lower dimensional structures, called
principle sets. Because the
ridge definition in \citet{Genovese2014} is very similar to the one of
principle curves and sets in \citet{OzertemETAL2011}, we will not state
it here. The differences are discussed after Definition \ref{def:ridge_hessian}.

A ridge as defined in \citet{Genovese2014} is a $s$-dimensional structure
containing all points, where the $s$ smallest eigenvalues of the Hessian
matrix are negative, and the corresponding eigenvectors are orthogonal to
the gradient; see Definition \ref{def:ridge_hessian}. This means, each
point in the ridge is a mode in the affine subspace spanned by the
eigenvectors corresponding to the $s$ smallest eigenvalues of the Hessian.

In our ridge definition (Definition \ref{def:ridge_variance}) we replace the
Hessian matrix by a conditional covariance matrix and instead of relaying on the
gradient to check for a mode in the affine subspace, we just require that
there is such a mode. Therefore, we can relax the conditions on the
probability density function; see Definition
\ref{def:ridge_variance}. However, if the stronger conditions hold, both
definitions are equivalent.

\subsection{Ridge definitions.}
\label{sec:properties-ridge}

\paragraph{Density function.}
We always assume that the density function $f \in \cC(\R^d)$. To show the
equivalence between both ridge definitions we assume further

(A1) $f \in \cC^2(\R^d)$ with
positive eigengap
\begin{displaymath}
  \delta(\x) \defeq \lambda_s(D^2\ell(\x)) - \lambda_{s+1}(D^2\ell(\x)) > 0
\end{displaymath}
for all $\x \in \{f > 0\}$, where $\ell = \log f$.

\paragraph{The kernel function.} All algorithm involve a bounded kernel
function $K:R^d \rightarrow [0, \infty)$ such that
\begin{itemize}
\item[(K0)] $\displaystyle\int K(\z) \, d\z = 1$.
  \item[(K1)] $K$ is sign- and permutation-symmetric, i.e.
    \begin{displaymath} K(z_1, z_2, \ldots, z_d) = K(\xi_1 z_{\sigma(1)},
      \xi_2 z_{\sigma(2)}, \ldots, \xi_d z_{\sigma(d)})
    \end{displaymath} for all $\z \in \R^d$, $\bs{\xi} \in \{-1, 1\}^d$
    and $\sigma \in \mathcal{S}_d$, the set of permutations on $\{1, 2,
    \ldots, d\}$.
  \item[(K2)] $\displaystyle\int K(\z) \z \z^\top = \I_d$, the identity
    matrix in $\R^d$.
\end{itemize}
To get equivalence between both ridge definitions or to apply the
algorithms presented later, we will require additional conditions.
\begin{itemize}
\item[(K3)] It holds
    \begin{displaymath}
    \mu_4 \defeq \int K(\z) z_1^4 \, d\z = 3 \quad \text{and} \quad \mu_{22}
    \defeq \int K(\z) z_1^2 z_2^2 \, d\z = 1.
  \end{displaymath}
  \item[(K4)] $K$ is rotationally symmetric, i.e.~there exists a
  \emph{profile} of the kernel $k:\R \rightarrow [0,\infty)$ such that
  $K(\z) = c_{k,d} k\bigl( \lVert \z \rVert^2 \bigr)$ for some constant
  $c_{k,d}$.
\item[(K5)] The kernel $K$ is log-concave and $K \in \cC^2_{b}(\R^d)$, this
  means, all partial derivatives of order 2 exist and are
  bounded. Moreover, the largest eigenvalue of $D^2\log K$ is negative and
  bounded away from $0$.
\end{itemize}
  
For our definitions and algorithms we need rescaled versions of $K$. For
any bandwidth $h > 0$ we write
\begin{displaymath}
  K_h(\z) \defeq h^{-d}K(h^{-1}\z).
\end{displaymath}
A particular choice of $K$ fulfilling (K0--5) is the standard Gaussian density
\begin{displaymath}
  \z \mapsto (2\pi)^{-d/2} \exp(-\lvert \z \rvert^2 / 2).
\end{displaymath}

\paragraph{Two density ridge definitions.}
Both definitions and algorithms involve the spectral decomposition of
either the Hessian matrix of $f$ or $\ell$, or the conditional covariance
matrix.

For a symmetric matrix $\M \in \R^{d \times d}$ we denote the vector of
eigenvalues in decreasing order by
\begin{displaymath}
  \blambda(\M) = (\lambda_1(\M), \lambda_2(\M), \ldots, \lambda_d(\M)) \quad
  \text{with} \quad \lambda_1(\M) \geq \lambda_2(\M) \ge \cdots \ge
  \lambda_d(\M), 
\end{displaymath}
the matrix with the eigenvalues on the diagonal by
$\bLambda(\M) = \diag\bigl( \blambda(\M) \bigr)$ and the corresponding
eigenvectors by
\begin{displaymath}
\V(\M) = [\v_1(\M), \v_2(\M), \ldots, \v_d(\M)], \quad \text{where}
 \quad\lambda_j(\M) \v_j(\M) = \M \v_j(\M) \quad \text{for} \quad 1 \leq j
 \leq d .
\end{displaymath}
We are usually interested in the space spanned by the $s$ largest or the
$d-s$ smallest eigenvalues. Hence we write for fixed $1 \leq s < d$
\begin{displaymath}
\V_\parallel(\M) = [\v_1(\M), \ldots, \v_s(\M)] \quad \text{and} \quad
\V_{\!\!\perp}(\M) = [\v_{s + 1}(\M), \ldots, \v_d(\M)].
\end{displaymath}
The matrix $\Vo(\M) \Vo(\M)^\top$ projects a vector onto the space
spanned by $\v_{s+1}(\M), \ldots, \v_{d}(\M)$. The distance between two
such subspaces generated by matrices $\A, \B \in \R^{d \times
  d}_{\text{sym}}$ is defined by
\begin{displaymath}
  \text{dist}\bigl( \Vo(\A), \Vo(\B) \bigr) \defeq \lVert \Vo(\A)
  \Vo(\A)^\top - \Vo(\B) \Vo(\B)^\top \rVert_F,
\end{displaymath}
where $\lVert \cdot \rVert_F$ denotes the Frobenius norm.

Next, we present a new ridge definition rely on fewer assumptions on
the underlying density function, followed by the usual ridge definition
given e.g~in \citet{Eberly1996}, \citet{Genovese2014}.

\begin{defi} \label{def:ridge_variance} Suppose (K0--3) hold. Let $\X$ and
  $\Z_h$ be independent random vectors with density functions $f$ and
  $K_h$, respectively. We define the \emph{conditional covariance matrix}
  \begin{displaymath}
    \bs{\Sigma}_h(\x) \defeq \Var\bigl( \X \big\vert \X + \Z_h = \x \bigr).
  \end{displaymath}
  Assume there exists a matrix $\V_{\!\!\perp}(\x) \in \R^{d \times (d-s)}$ with
  orthonormal columns, such that
  \begin{displaymath}
    \text{dist}\bigl(\Vo(\bs{\Sigma}_h(\x)),
    \Vo(\x) \bigr) \rightarrow 0 \quad \text{as} \quad h
    \rightarrow \infty \quad \text{for each } \x \in \{f > 0\}.
  \end{displaymath}
  The \emph{$s$-dimensional ridge} of $f$ is then
  \begin{displaymath}
    R_s(f) \defeq \bigl\{ \x \in \R^d: \R^{d-s} \ni \z \mapsto f(\x + \V_{\!\!\perp}(\x) \z) \ \text{has a mode at} \ \bs{0}_{d-s} \bigr\}.
  \end{displaymath}
\end{defi}

\begin{defi} \label{def:ridge_hessian}
  Let $f \in \cC^2(\R^d)$ be a probability density function with gradient
  $\bs{g}(\x)$ and Hessian matrix $\bs{H}(\x)$ at point $\x$. 
  The \emph{ridge $\tilde{R}_s(f)$ of $f$ with dimension $s$} is given by
  \begin{displaymath}
    \tilde{R}_s(f) = \{\x : \V_{\!\!\perp}(\bs{H}(\x))^\top \bs{g}(\x) = \bs{0},
    \lambda_{s+1}(\bs{H}(\x)) < 0\}.
  \end{displaymath}
\end{defi}

Definition \ref{def:ridge_hessian} is almost identical to the definition of
principle sets in \citet{OzertemETAL2011}. The only
difference is in the inclusion or exclusion of points already contained in a
lower dimensional ridge. In Definition \ref{def:ridge_hessian} we have
\begin{displaymath}
  \tilde{R}_0(f) \subset \tilde{R}_1(f) \subset \ldots \subset
  \tilde{R}_{d-1}(f).
\end{displaymath}
However, a principle set of dimension $s$ is defined as $\tilde{R}_0(f)$
for $s = 0$ and $\tilde{R}_s(f) \backslash \tilde{R}_{s-1}(f)$ for
$1 \leq s \leq d-1$. Hence, in general, only the 0-dimensional principle
sets and ridges coincide, which are the local maxima of the probability
density function.

\citet{OzertemETAL2011} show that the principle sets, and hence the ridges,
are the same for $f$ and $p \circ f$, where $p$ is a monotonically
increasing function on $\R$. Especially, replacing the Hessian of $f$ with the Hessian of
$\log \circ f$ in Definition \ref{def:ridge_hessian} leads to the same ridge.

\paragraph{Weighted distribution.}
Let $\X$ and $\Z_h$ be independent with probability density $f$ and $K_h$,
respectively. Let $q$ be the probability density of $\X$, given
$\X + \Z_h = \x$, then
\begin{displaymath}
  q(\y) = \frac{f_{\X, \X + \Z_h = \x}(\y, \x)}{f_{\X + \Z_h}(\x)} =
  \frac{f(\y) K_h(\y-\x)}{\int f(\z) K_h(\z - \x)\, d\z} = s_h(\x)^{-1}
  K_h(\y - \x) f(\y),
\end{displaymath}
where $s_h(\x) \defeq \int K(\z) f(\x + h\z) \, d\z$. Let the conditional
expectation be
\begin{align*}
  \muhx
  & \defeq \Ex(\X \, \vert \, \X + \Z_h = \x)\\
  & = s_h(\x)^{-1} \int K(\z) (x + h\z) f(\x + h\z) \, d\z\\
  & = \x + h s_h(\x)^{-1} \s_h(\x),
\end{align*}
where $\s_h(\x) \defeq \int K(\z) \z f(\x + h\z) \, d\z$ and the conditional
variance
\begin{align*}
  \Sigmahx
  & \defeq \Var(\X \, \vert \X + \Z_h = \x)\\
  & = \int \y \y^\top q(\y) \, d\y - \muhx \muhx^\top\\
  & = h^2 \bigl( s_h(\x)^{-1} \S_h(\x) - s_h(\x)^{-2} \s_h(\x)
    \s_h(\x)^\top \bigr), 
\end{align*}
 because
\begin{align*}
  \int \y \y^\top f(\y) K_h(\y - \x) \, d\y
   &= \int (\x + h\z) (\x + h\z)^\top K(\z) f(\x + h\z) \, d\z\\
   &= \x \x^\top s_h(\x) + h \x \s_h(\x)^\top + h \s_h(\x) \x^\top + h^2 \S_h(\x),
\end{align*}
where $\S_h(\x) \defeq \int K(\z) \z \z^\top f(\x + h\z) \, d\z$.

The following Lemma is a direct consequence of Corollary 2.2 in
\citet{StraehlETAL2020} about local moments.
\begin{lem}
  \label{lem:conditional-distribution}
  Suppose $f \in \cC^2(\R^d)$ and conditions (K0--2) hold, then for $\x \in
  \{f > 0\}$, 
  \begin{align*}
    \muhx &= \x + h^2f(\x)^{-1} Df(\x) + \o(h^3),\\
    \Sigmahx & = h^2 \I_d + 2^{-1}(\mu_{22} - 1) h^4 f(\x)^{-1} \tr\bigl(
               D^2f(\x) \bigr) \I_d + h^4 f(\x)^{-1} D^2f(\x) \odot \M\\
    & \quad - h^4 f(\x)^{-2} Df(\x) Df(\x)^\top + \o(h^4)
\end{align*}
as $h \to 0$, locally uniformly in $\x$. Here $\odot$ denotes the
componentwise product of matrices, and $\M$ is the matrix with entries
$M_{jk} := \one\{j=k\}(\mu_4 - \mu_{22})/2 + \one\{j\ne k\} \mu_{22}$ for
$j,k = 1,\ldots,d$. If the kernel function fulfills (K3), then
\begin{displaymath}
  \Sigmahx = h^2 \I_d + h^4 D^2 \log f(\x) + \o(h^4) \quad \text{as} \quad
  h \rightarrow 0, \text{ locally uniformly in } \x.
\end{displaymath}
\end{lem}

\paragraph{Equivalence of the two ridge definitions.}
Theorem \ref{thm:hessian_variance} shows that under (A1) and (K0--3) the
subspace generated by the conditional covariance matrix converges to the
subspace generated by the Hessian of the log-density and Theorem
\ref{thm:hessian_variance2} shows that under the same assumptions at any ridge
point $\x \in R_s(f)$, the density function restricted to the subspace
spanned by $\Vo(D^2f(\x))$ is log-concave with mode at
$\bs{0}_{d-s}$. Together they imply $\tilde{R}_s(f) \subset
R_s(f)$. Theorem \ref{thm:hessian_variance3} shows $R_s(f) \subset
\tilde{R}_s(f)$ again under the same conditions, and so $\tilde{R}_(f) =
R_s(f)$. 

\begin{thm} \label{thm:hessian_variance} Suppose (A1) and (K0--2) hold, then
  \begin{displaymath}
    \dist\bigl( \Vo(\bSigma_h(\x)), \Vo(\bs{H}(\x)) \bigr) \rightarrow 0
    \quad \text{as} \quad h \rightarrow 0.
\end{displaymath}
\end{thm}

\begin{thm} \label{thm:hessian_variance2} Suppose (A1) and K(0--2) hold and
  $\x \in \tilde{R}_s(f)$; see Definition \ref{def:ridge_hessian}. Then there
  exists $\varepsilon > 0$, such that the function
  \begin{displaymath}
    \z'' \mapsto f(\x + \Vo(D^2f(\x))\z'')
  \end{displaymath}
is log-concave on $\{\z'' \in \R^{d-s}: \lVert \z'' \rVert < \varepsilon\}$
with a mode at $\bs{0}_{d-s}$. In particular, $t \mapsto f(\x + t\u)$ has a
mode at $0$ for any $\u$ in the column space of $\Vo(D^2\ell(\x))$.
\end{thm}


\begin{thm}
  \label{thm:hessian_variance3}
 Suppose that (A1) and (K0--2) hold and  $\x \in R_s(f)$; see Definition
  \ref{def:ridge_variance}. Then
  \begin{displaymath}
    \Vo(D\ell(\x))^\top D\ell(\x) = \bs{0} \quad \text{and} \quad
    \lambda_{s+1}(D^2\ell(\x)) < 0.
  \end{displaymath}
\end{thm}

\paragraph{Projected weighted distribution.}
Suppose $f \in \cC(\R^d)$ and the conditions (K0--4) hold. Let $\V = [\Vp,
\Vo] \in \R^{d \times d}$ be a orthogonal matrix with $\Vp \in \R^{d \times
s}$ and $\Vo \in \R^{d \times (d-s)}$. The
probability density function of $\X - \x$, given $\X + \Z_h
= \x$, is 
\begin{displaymath}
  \y \mapsto \frac{K_h(\y)f(\y + \x) }{\int K_h(\z)f(\z + \x) \, d\z} =
  \s_h(\x)^{-1} K_h(\y)f(\y + \x) .
\end{displaymath}
Hence, the probability density function of the rotated distribution
$\V^\top(\X - \x)$, given $\X + \Z_h = \x$, is
  \begin{displaymath}
    \y \mapsto \s_h(\x)^{-1} \bigl\lvert \det(\V)^{-1} \bigr\rvert
    K_h\bigl( \V^{-\top} \y \bigr) f\bigl(
    \x + \V^{-\top}\y \bigr)  =
    \s_h(\x)^{-1} K_h(\V \y) f(\x + \V \y) .
  \end{displaymath}
  Finally, the probability density function of $\V_{\!\!\perp}^\top (\X -
  \x)$, given $\X + \Z_h = \x$, is
  \begin{displaymath}
    \y'' \mapsto \s_h(\x)^{-1} \int_{\R^s} K_h(\y) f(\x + \V
    \y)  \, d\y',
  \end{displaymath}
where $\y = (\y', \y'')$. Note that $K_h(\V\y) = K_h(\y)$ by (K4). In the
case of $\Vo = \Vo\bigl( \Sigmahx \bigr)$ as in Definition \ref{def:ridge_variance}, we define
\begin{displaymath}
  g_h(\z'') \defeq \s_h(\x)^{-1} \int_{\R^s} K(\z) f\bigl(\x + h\V(
  \Sigmahx) \z\bigr)  \, d\z',
\end{displaymath} 
the weighted density of $f$ projected onto the $(d-s)$-dimensional space
spanned by the directions of lowest conditional variance.

\begin{thm}
  \label{thm:projected-density}
  Suppose $f = e^\ell$ with $\ell \in \cC^2(\R^d)$ and bounded second
  order partial derivatives. Suppose conditions (K0--5) hold and
  $\x \in R_s(f)$. Then there exists $h_o > 0$ such that the projected
  weighted density $g_h$ is log-concave for all $0 < h \le h_o$.
\end{thm}

Theorem \ref{thm:projected-density} justifies using a log-concave density
estimator on the sample projected weighted distribution to estimate a point
on the density ridge.

\section{Algorithms}
\label{sec:algorithm}

We consider independent random vectors $\X_1, \X_2, \ldots, \X_n$ with
distribution given by the density function $f : \R^d \rightarrow [0,
\infty)$. Our goal is to estimate the density ridge of $f$.

\subsection{Mean Shift Algorithm}
\label{sec:mean-shift-algorithm-1}

\begin{algorithm}[t]
  \DontPrintSemicolon
  \KwData{$\X_1, \X_2, \ldots, \X_n$}
  \KwIn{Starting point $\x$, bandwidth $h > 0$, $\text{tol} > 0$}
  \KwResult{Mode close to $\x$}
  \Begin{
    $\bs{m} \leftarrow \bs{m}_{h,G}(\x)$\;
    \While{$\lVert \bs{m} \rVert > \text{tol}$}{
      $\x \leftarrow \x + \bs{m}$\;
      $\bs{m} \leftarrow \bs{m}_{h,G}(\x)$\;
    }
    \Return{$\x$}
  }  
  \caption{Mean Shift \label{alg:MeanShift}}  
\end{algorithm}

The mean shift algorithm is an iterative procedure to find the modes of a
distribution; see \citet{Cheng1995}, \citet{ComaniciuMeer2002}. First, the
density is estimated by kernel density estimation (KDE) with a
rotationally symmetric kernel $K$ as in (K4). An estimator
for $f$ and $Df$ is then given by
\begin{displaymath}
  \hat{f}_{h,K}(\x) = \frac{c_{k,d}}{nh^d} \sum_{i = 1}^n k \Bigl( \Bigl\lVert
  \frac{\X_i - \x}{h} \Bigr\rVert^2 \Bigr) \quad \text{and} \quad D
  \hat{f}_{h,K}(\x) = - \frac{2c_{k,d}}{nh^{d+2}} \sum_{i=1}^n (\X_i - \x)
  k'\Bigl( \Bigl\lVert \frac{\X_i - \x}{h}  \Bigr\rVert^2 \Bigr),
\end{displaymath}
respectively. By defining $g(\x) = - k'(\x)$ and $G(\x) = c_{g,d} g\bigl(
\lVert \x \rVert^2 \bigr)$, where $c_{g,d}$ is the corresponding
normalization constant, we can write
\begin{align*}
  D \hat{f}_{h,K}(\x)
  & = \frac{2c_{k,d}}{nh^{d+2}} \left( \sum_{i=1}^n g
    \Bigl( \Bigl\lVert \frac{\X_i - \x}{h}  \Bigr\rVert^2 \Bigr) \right)
  \left( \frac{\sum_{i=1}^n \X_i g\Bigl( \bigl\lVert \frac{\X_i - \x}{h}
      \bigr\rVert^2 \Bigr)}{\sum_{i=1}^n g\Bigl( \bigl\lVert \frac{\X_i - \x}{h}
      \bigr\rVert^2 \Bigr)} - \x \right)\\
  & = \hat{f}_{h,G}(\x) \frac{2c_{k,d}}{h^2c_{g,d}} \bs{m}_{h,G}(\x),
\end{align*}
with
\begin{displaymath}
  \hat{f}_{h,G}(\x) \defeq \frac{c_{g,d}}{nh^d} \sum_{i=1}^n g \Bigl(
  \Bigl\lVert \frac{\X_i - \x}{h} \Bigr \rVert^2 \Bigr) \quad \text{and}
  \quad \bs{m}_{h,G}(\x) \defeq \frac{\sum_{i=1}^n \X_i g\Bigl( \bigl\lVert \frac{\X_i - \x}{h}
      \bigr\rVert^2 \Bigr)}{\sum_{i=1}^n g\Bigl( \bigl\lVert \frac{\X_i - \x}{h}
      \bigr\rVert^2 \Bigr)} - \x,
\end{displaymath}
the KDE with kernel $G$ and the \emph{mean shift} $\bs{m}_{h,G}$,
respectively. A new candidate $\x_{\rm new}$ for the mode ideally satisfy
$D\hat{f}_{h,K}(\x_{\rm new}) = 0$, or equivalently $\bs{m}_{h,G}(\x_{\rm
  new}) = 0$, whenever $\hat{f}_{h,G}(\x_{\rm new}) > 0$. We mimic this by
setting $\x_{\rm new} = \x + \bs{m}_{h,G}(\x)$. This leads to Algorithm \ref{alg:MeanShift}.

If $K$ has a monotonically decreasing profile $k$, the sequence of $\y$ and
$\hat{f}_{h,K}(\y)$ converge and the latter is monotonically
increasing.
For the Gaussian-kernel we have
\begin{displaymath}
  k(y) = \exp\bigl( -y / 2 \bigr) \quad \text{and} \quad K(\y) =
  (2\pi)^{-d/2} \exp\bigl( - \lVert \y \rVert^2 / 2 \bigr),
\end{displaymath}
with profile
\begin{displaymath}
  g(y) = \frac{1}{2} \exp\bigl( - y / 2 \bigr) \quad \text{and} \quad
  G(\y) = K(\y).
\end{displaymath}
For given data the number of steps for convergence depends on the
chosen kernel. If $G$ is the uniform kernel the number of steps are finite,
otherwise the algorithm should be stopped if the length of the mean shift
vector is below a certain threshold; see \citet{ComaniciuMeer2002}.

\subsection{Subspace Constraint Mean Shift Algorithm}
\label{sec:subsp-constr-mean}

\begin{algorithm}[t]
  \DontPrintSemicolon
  \KwData{$\X_1, \X_2, \ldots, \X_n$}
  \KwIn{Starting point $\x$, bandwidth $h > 0$, $\text{tol} > 0$}
  \KwResult{Ridge point close to $\x$}
  \Begin{ $\bs{m} \leftarrow \bs{m}_{h,G}(\x)$\;
    $\bs{g} \leftarrow D\hat{f}_{h,K}(\x)$\;
    $\bs{H} \leftarrow D^2\log \hat{f}_{h,K}(\x)$\;
    \While{$\bs{\lvert \bs{g}^\top \bs{H} \bs{g} \rvert} > (1 - \mbox{tol})
      \lVert \bs{g} \rVert \cdot \lVert \bs{H} \bs{g} \rVert$}{
      $\x \leftarrow \x + \Vo(\bs{H}) \Vo(\bs{H})^\top
      \bs{m}$\; 
      $\bs{m} \leftarrow \bs{m}_{h,G}(\x)$\;
      $\bs{g} \leftarrow D\hat{f}_{h,K}(\x)$\;
      $\bs{H} \leftarrow D^2\log \hat{f}_{h,K}(\x)$\;
    }
    \Return{$\y$}
  }
  \caption{Subspace Constraint Mean Shift \label{alg:SCMS}}
\end{algorithm}

The subspace constraint mean shift algorithm (SCMS) is a modification of
the mean shift algorithm. It was first proposed by \citet{OzertemETAL2011}.
We move in direction of a projected mean shift vector to find a ridge
point, where we project onto the space spanned by the $d-s$ eigenvectors,
corresponding to the $d-s$ largest eigenvalues of the estimated negative
Hessian of the log-density, called the \emph{local
  covariance-inverse}. Replacing the Hessian of the density with the
Hessian of the log-density does not change the ridge; see
\citet{OzertemETAL2011}.

Using KDE leads to the matrix
\begin{displaymath}
  - \frac{D^2\hat{f}_{h,K}(\x)}{\hat{f}_{h,K}(\x)} + \frac{D
    \hat{f}_{h,K}(\x) D \hat{f}_{h,K}(\x)^\top}{\hat{f}_{h,K}(\x)^2} =
  - D^2\log \hat{f}_{h,K}(\x).
\end{displaymath}

We will use the positive Hessian of the log-density, hence we project on
the $d-s$ eigenvectors corresponding to the $d-s$ largest eigenvalues.
The procedure is explained in Algorithm \ref{alg:SCMS}.

\paragraph{Bias.}
The SCMS algorithm finds the ridge points of the underlying kernel density
estimator, this leads to a ridge estimation of $K_h * f$ instead of $f$. If
the kernel fulfills (K0--2), then this is asymptotically
\begin{displaymath}
  \hat{f}_n(\x) = f(\x) + \O(h^2) + \O_p(n^{-1/2}h^{-d/2}) \quad
  \text{uniformly in } \x \in \R^d,
\end{displaymath}
see also Section \ref{sec:DKE-asym}. In \citet{Genovese2014} it is shown, that
under some regularity conditions on $f$, we have
\begin{displaymath}
  \text{Haus}\bigl( R(f), R(K_h * f) \bigr) = \O(h^2),
\end{displaymath}
where
\begin{displaymath}
  \text{Haus}(A, B) \defeq \inf\{\delta : A \subset B \oplus \delta \text{
    and } B \subset A \oplus \delta\}
\end{displaymath}
is the Hausdorff distance with
\begin{displaymath}
S \oplus \delta \defeq \{\bs{s} + \z:
\bs{s} \in S, \z \in \R^d \ \text{with} \  \lVert \z \rVert \leq \delta\} \quad
\text{for} \quad S \subset \R^d.
\end{displaymath}
Hence, the bias of the density estimation also
effects the ridge estimation.

\paragraph{Uncertainty measure.}
The local uncertainty measure defined in \citet{WassermanETAL2015} is
given by
\begin{displaymath}
  \rho(\x)^2 \defeq
  \begin{cases}
    \Ex(d^2(\x,\hat{R})), & \text{if } \x \in R(K_h * f),\\
    0, & \text{otherwise,}
  \end{cases}
\end{displaymath}
where $\hat{R}$ is the estimated ridge and $d(\x, A) \defeq \min\{\lVert \x
- \y\rVert : \y \in A\}$ for any compact $A \subset \R^d$. It can be used
for showing the uncertainty of an estimated ridge point, unfortunately, it
only takes into account the variance part but not the bias part.

\citet{WassermanETAL2015} also developed an algorithm for estimating
$\rho(\x)$ based on bootstrap samples and show consistency thereof; see
\citet[Theorem 5]{WassermanETAL2015}. Moreover, they also show consistency
for a bootstrap confidence set for the smoothed ridge $R(K_h * f)$.

\subsection{Log-Concave Ridge Search Algorithm}
\label{sec:vari-constr-mode}

\begin{algorithm}[t]
  \DontPrintSemicolon
  \KwData{$\X_1, \X_2, \ldots, \X_n$}
  \KwIn{Starting point $\x$, bandwidth $h > 0$, $\text{tol} > 0$}
  \KwResult{Ridge point of $f$ close to $\x$}
  \Begin{$m \leftarrow 2 \ \text{tol}$\;
    \While{$\lvert m \rvert > \text{tol}$} {$\bs{H} \leftarrow \S_n(\x) /
      s_n(\x) - \s_n(\x) \s_n(\x)^\top / s_n(\x)^2$\;
      $\v \leftarrow \V_{\!\!\perp}(\H)$\;
      $\w_i \leftarrow s_n(\x)^{-1} K_h(\X_i - \x)$ for $1 \leq i \leq n$\;
      $z_i \leftarrow \v^\top(\X_i - \x)$ for $1 \leq i \leq n$\;
      $m \leftarrow \mode\bigl( \hat{\theta}(\z, \bs{w})\bigr)$\;
      $\x \leftarrow \x + m \v$\;
    }
    \Return{$\x$}
  }
  \caption{Log-Concave Ridge Search}
\end{algorithm}

The log-concave ridge search (LCRS) is based on Definition
\ref{def:ridge_variance}. For a starting point $\x$ we look for the
direction with smallest weighted variance, project the weighted data onto
the (affine) subspace spanned by those directions and iterate to the mode
of the log-concave density estimated from the weighted projected data. We
repeat this step until the step-size is below a chosen threshold. The
algorithm can be used to find $(d-1)$-dimensional ridges.

\paragraph{Finding Direction.} For a point $\x$ we chose the direction of
smallest conditional variance. Therefore, consider the empirical measure
\begin{displaymath}
  \hat{Q}_{\x,n,h} \defeq \sum_{i=1}^n w_i(\x) \delta_{\X_i - \x}, \quad
  \text{where} \quad  w_i(\x) = s_{n,h}(\x)^{-1} n^{-1} K_h(\X_i - \x)
\end{displaymath}
with $s_{n,h}(\x) \defeq n^{-1} \sum_{i=1}^d K_h(\X_i - \x)$. The empirical conditional variance is then
\begin{align*}
  \hat{\bs{\Sigma}}_{n,h}(\x)
  &  \defeq \Var(\hat{Q}_{\x, n, h})\\
  &= \sum_{i=1}^n w_i(\x) (\X_i - \x)(\X_i - \x)^\top
    - \bigl( \sum_{i=1}^n w_i(\x) (\X_i - \x)\bigr) \bigl( \sum_{i=1}^n
    w_i(\x) (\X_i - \x)\bigr)^\top\\
  & = h^2 \left(\frac{\bs{S}_{n,h}(\x)}{s_{n,h}(\x)} - \frac{\s_{n,h}(\x)
    \s_{n,h}(\x)^\top}{s_{n,h}(\x)^2}\right), 
\end{align*}
with
\begin{align*}
  \s_{n,h}(\x) & \defeq \frac{1}{n} \sum_{i=1}^n K_h(\X_i - \x)h^{-1} (\X_i -
                 \x),\\
  \S_{n,h}(\x) & \defeq \frac{1}{n} \sum_{i=1}^n K_h(\X_i - \x)h^{-2} (\X_i -
                 \x)(\X_i - \x)^\top.
\end{align*}
We have
\begin{displaymath}
  \hat{\bSigma}_{n,h}(\x) = h^2 \I_d + h^4 D^2 \ell(\x) + \o(h^4) + \O_p(n^{-1/2}h^{-d/2-2})
\end{displaymath}
and by Lemma \ref{lem:conditional-distribution} we get
\begin{displaymath}
  \hat{\bs{\Sigma}}_{n,h}(\x) - \Sigmahx = \o(h^4) + \O_p(n^{-1/2}h^{-d/2 -
  2}) \quad \text{as} \quad h \rightarrow 0.
\end{displaymath}
Hence we have a consistent estimator of the direction of smallest
conditional variance, whenever $nh^{d+4} \rightarrow \infty$. This
direction is then the eigenvector of $\hat{\bs{\Sigma}}_{n,h}(\x)$
corresponding to the smallest eigenvalue, this is
$\v \defeq \Vo\bigl( \hat{\bs{\Sigma}}_{n,h}(\x) \bigr)$. The following
theorem shows, that a small perturbation of $\x$ only leads to a small
perturbation of $\hat{\bs{\Sigma}}_{n,h}(\x)$, whence only to a small
perturbation of $\v$.

\begin{thm}[Lipschitz Continuity]
  \label{thm:Lipschitz}
  Let $f \in \cC^2(\R^d)$ and suppose (K0--2) and (K5) hold.  For a sample
  $\mathcal{X} = \{\X_1, \ldots, \X_n\}$ and fixed $h$, the local sample
  variance is Lipschitz continuous on the convex hull of the sample, i.e.
  \begin{displaymath}
    \lVert \hat{\bs{\Sigma}}_{n,h}(\x) - \hat{\bs{\Sigma}}_{n,h}(\y) \rVert_F \leq
    L \lVert  \x - \y \rVert \quad \text{for} \ \x, \y \in \conv(\X_1, \ldots,
    \X_n),
  \end{displaymath}
  whenever there exists $0 < \tau \leq s_n(\x)$ for all $\x \in \conv(\X_1,
  \ldots, \X_2)$, with some $L > 0$ depending on the sample, the kernel $K$ and
  bandwidth $h$.
\end{thm}

\paragraph{Finding mode of projection.} For a point $\x$ and the direction
$\v$ we define
\begin{displaymath}
  z_i = \v^\top(\X_i - \x) \quad \text{with weights} \quad w_i =
  s_n(\x)^{-1} K_h(\X_i -
  x) \quad  \text{for} \ 1 \leq i \leq n.
\end{displaymath}
These leads the empirical measure
\begin{equation}
  \label{eq:emprical-projected-measure}
  \hat{P}_{\x,n,h} \defeq \sum_{i=1}^n w_i \delta_{z_i}.
\end{equation}
We use the maximum likelihood estimation for log-concave distributions
proposed in \citet{DuembgenRufibach2009}. The algorithm is explained in
\citet{ DuembgenRufibach2011} and refined in \citet{DuembgenETAL2018}. The
estimated log-density $\hat{\theta}(\z, \w)$ with $\z = (z_1, \ldots, z_n)$
and $\w = (w_1, \ldots, w_n)$ is piecewise linear with change of slope at
data points, convex and unique. Hence, the algorithm will return a unique
mode $m$ almost surely. We update the considered point to $\x + m \v$.

\paragraph{Interpretation of direction.}
\label{sec:direction}

A $\log$-transformation of the density does not change the ridge
set. However, in case of a Gaussian distribution with covariance matrix
$\bs{\Sigma}$, the Hessian of the $\log$-density is independent of
location and equal to
\begin{displaymath}
  D^2\log f(\x) = - 2^{-1} \bs{\Sigma}^{-1}.
\end{displaymath}
The $d-s$ eigenvectors associated to the $d-s$ largest eigenvalues coinside
with the $d-s$ last linear principle components of the
distribution. Hence, using the $\log$-transformation gives a beneficial
interpretation of the considered subspace. 

The conditional covariance matrix $\Var(\X \, | \, \X + \Z_h = \x)$ leads
the same interpretation. For a Gaussian distribution $\X$, the conditional
distribution of $\X$, given $\X + \Z_h = \x$ with $\Z_h \sim \NN(\bs{0}, h^2
\I_d)$ for $h > 0$ is Gaussian as well and the considered subspace
coincides with the $d-s$ last linear principle components. Another
interpretation of the conditional covariance matrix is as the
$(d-s)$-dimensional subspace with conditional least variance. Hence, the ridge is
in direction of largest variance.

\paragraph{Bandwidth selection.}
One crucial part of the algorithms is selecting the bandwidth. Whereas in
the calculation of the mode via log-concave density estimation, the result
does not change drastically for different bandwidths, it can have an effect
on the conditional variance and the resulting direction of smallest
conditional variance. Therefore, we will focus on a suitable bandwidth for
the latter.

In \citet{Chen2015} they recommend choosing $h$ via
\begin{equation}
  \label{eq:Silverman}
  h = A_0 (d+2)^{-1/(d+4)} n^{-1/(d+4)} \sigma_{\text{min}},
\end{equation}
where $A_0$ is some constant, $d$ is the dimension and $\sigma_{\text{min}}$ is
the minimal value for the standard deviation along each coordinate. For
$A_0 = 1$, one obtains Silverman's rule; see \citet{Silverman1986}. 

Another choice is using a functional of the length of the euclidean minimal spanning tree
(EMST) as bandwidth. For a sample $\X_1, \ldots, \X_n$ consider the
fully connected, undirected graph $G = (V, E)$ with vertices
$V = \{\X_1, \ldots, \X_n\}$ and edges
$E = \{(\X_i, \X_j): 1 \leq i < j \leq n\}$. The EMST is the graph
$E_n \defeq E_n(\X_1, \ldots, \X_n) \subset G$ that connects all vertices
in $V$ such that the total edge length is minimized. Let $L_n$ be the
length of $E_n(\X_1, \ldots, \X_n)$, then we chose the bandwidth as 
\begin{equation}
  \label{eq:EMST}
  h_n = \Bigl( \frac{L_n}{n} \Bigr)^{1 / (d + 4)}.
\end{equation}
In \citet{SreevaniMurthy2016} they use
\begin{displaymath}
  T_n = \Bigl( \frac{L_n}{n} \Bigr)^{1/d}
\end{displaymath}
as a bandwidth for the kernel density estimator and show that $L_n
\rightarrow \infty$, $h_n \rightarrow 0$ and $n T_n^d \rp \infty$ as $n
\rightarrow \infty$ under some mild conditions on the kernel and the
density function, the most restrictive being compact support of the density
function. From those two results we get immediately, that $nh_n^{d+4} \rp \infty$
as $n \rightarrow \infty$, the desired rate for estimating $D^2\log f(\x)$
consistently in case of $f \in \cC^2(\R^d)$. Because it is
\begin{align*}
  h^{-4}\bs{\hat{\Sigma}}_{n,h}(\x) - D^2 \log f(\x)
  & = h^{-4}\bigl(\bs{\hat{\Sigma}}_{n,h}(\x) - \Sigmahx  \bigr) + h^{-4}
    \Sigmahx - D^2 \log f(\x)\\
  & = \O_p(n^{-1/2}h^{-d/2 - 2}) + \o(1)\\
  & = \o_p(1) \quad \text{if } nh^{d+4} \rightarrow \infty.
\end{align*}

\paragraph{Confidence region of the ridge.}
\label{sec:conf-regi-mode}
For each point $\x$ on the estimated ridge one can calculate a confidence
interval along the direction $\v \defeq \Vo\bigl(
\hat{\bs{\Sigma}}_{n,h}(\x) \bigr)$ by using the likelihood ratio test
suggested in \citet{DossWellner2019}.

Let $\mathcal{P}$
be the family of all log-concave densities on $\R$. Suppose $g = e^\varphi \in
\mathcal{P}$, where $\varphi$ has second derivative $\varphi''$ at the mode
$m(g)$  and satisfies $\varphi''(m) < 0$. Consider the following testing
problem: $H_o: m(g) = m$ versus $H_1: m(g) \neq m$, where $m \in \R$ is
fixed. One can then calculate the unconstrained maximum likelihood
estimator (MLE) $\hat{g}_n$ and the
mode-constrained MLE $\hat{g}^o_n$ with $m(\hat{g}^o_n) = m$. Those two
functions lead to the log-likelihood statistic, given by
\begin{displaymath}
  2 \log \lambda_n = 2 \log \lambda_n(m) = 2n \mathbb{P}_n(\log \hat{g}_n - \log \hat{g}^o_n) =
  2n \mathbb{P}_n(\hat{\varphi}_n - \hat{\varphi}^o_n),
\end{displaymath}
where $\hat{\varphi} = \log \hat{g}_n$, $\hat{\varphi}^o = \log \hat{g}^o_n$,
$\mathbb{P}_n = \sum_{i = 1}^n w_i \delta_{X_i}$, and $\mathbb{P}_n(q) =
\int q \, d\mathbb{P}_n$.

\begin{thm}[Theorem 1.1 in \citet{DossWellner2019}]
  If $X_1, X_2, \ldots, X_n$ are i.i.d. $g = e^\varphi$ with mode $m$, where
  $\varphi$ is concave, twice continuously differentiable at $m$, and
  $\varphi''(m) < 0$, then
  \begin{displaymath}
    2 \log \lambda_n \rd \mathbb{D},
  \end{displaymath}
where $\mathbb{D}$ is a universal limiting distribution.
\end{thm}

Let $c_\alpha$ be such that $\Pr(\mathbb{D} > \alpha) = \alpha$, where 
the distribution of $\mathbb{D}$ can be approximated by Monte Carlo
methods. We can reject $H_0$ at level $\alpha$, whenever $2 \log \lambda_n
> c_\alpha$. A asymptotic $\alpha$-confidence interval for the mode is then
given by 
\begin{displaymath}
  I_{\alpha}(\mathbb{P}_n) = \{m \in \R: 2 \log \lambda_n(m) \le c_\alpha \}.
\end{displaymath}
We apply this procedure for the measure $\hat{P}_{\x,n,h}$ for $\x$ on the
estimated ridge and get the one-dimensional confidence region
\begin{displaymath}
  \{\x + t\v : t \in I_{\alpha}(\hat{P}_{\x,n,h})\}.
\end{displaymath}
Of course, this doesn't lead by any means to a confidence region. However,
it can be seen as a measure of uncertainty.

\paragraph{Threshold intervals for the ridge.}
\label{sec:interval-mode}
An issue of the LCRS is whenever the projected weighted density function is
rather flat close to the mode, a small change of the initial point may lead
to a different estimated mode in the next step of the algorithm, and whence
to ridge points far away from each other. In Figure \ref{fig:LCRS-jump} we
see two different, but very close, starting points leading to
different modes in the first step. After 1 and 8 steps, respectively, the
algorithm stops for both starting points at different modes.  However,
looking at the threshold interval, they are very close to each other; see
Figure \ref{fig:LCRS-jump-end} and \ref{fig:jump-interval}.

Instead of only reporting the estimated
ridge point, we use the estimated log-concave
density $\hat{\theta}$ to find the interval
\begin{displaymath}
  I(\hat{\theta}) \defeq \{x \in \R: \hat{\theta}(x) \geq m +
  \log(\tau)\} = \{x
  \in \R : e^{\hat{\theta}(x)} \geq \tau e^{\hat{\theta}(m)}\}.
\end{displaymath}
Thus, we get a uncertainty measure for the ridge. The threshold interval is
computationally much less expensive then the confidence intervals.

\begin{figure}[h]
  \begin{subfigure}[b]{0.5\textwidth}
    \includegraphics[width=\textwidth]{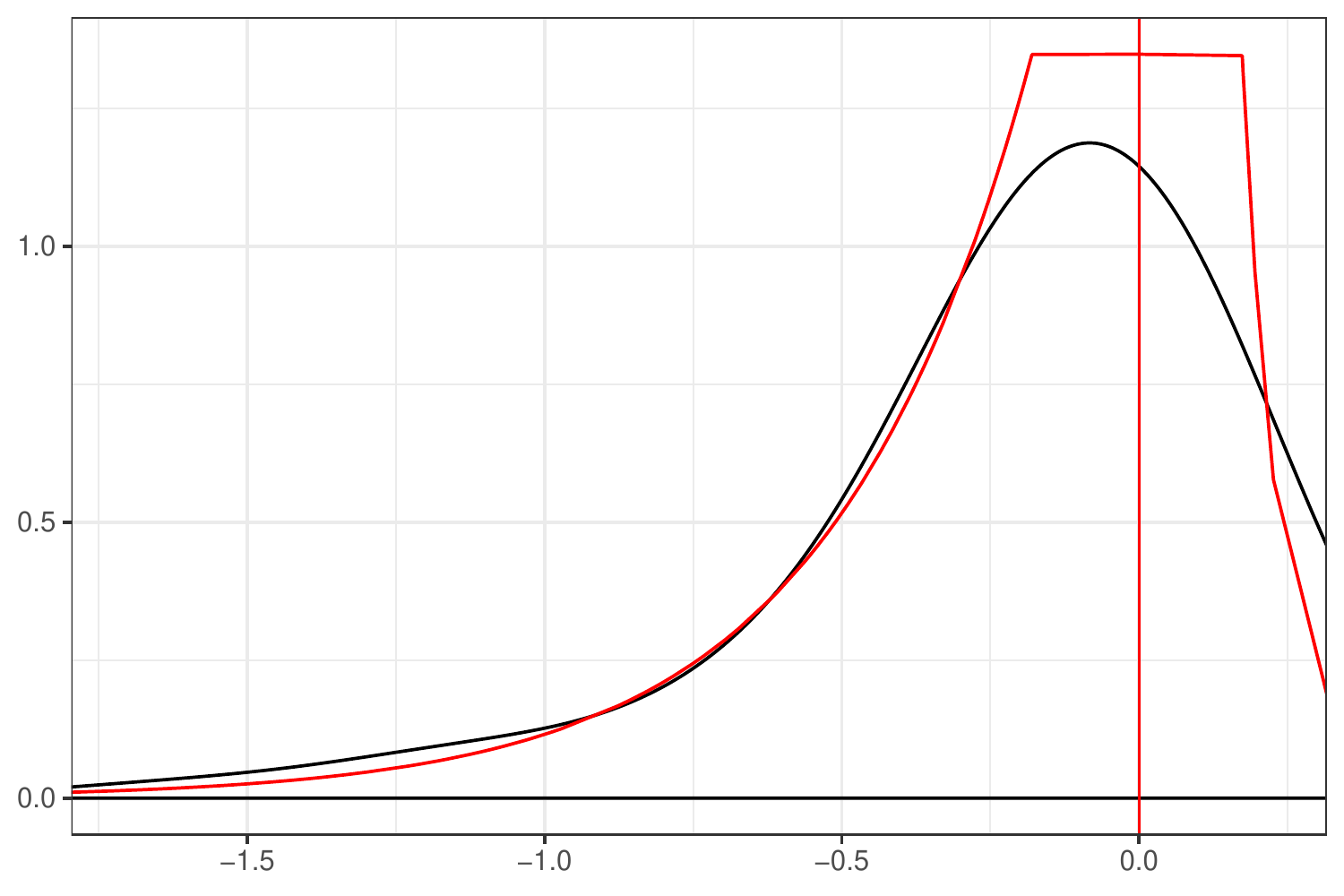}
    \caption{$\x = (-0.886, 0.2586)$, $\v = (-0.9475, 0.3199)$}
    \label{fig:rough-1}
  \end{subfigure}
  \hfill
  \begin{subfigure}[b]{0.5\textwidth}
    \includegraphics[width=\textwidth]{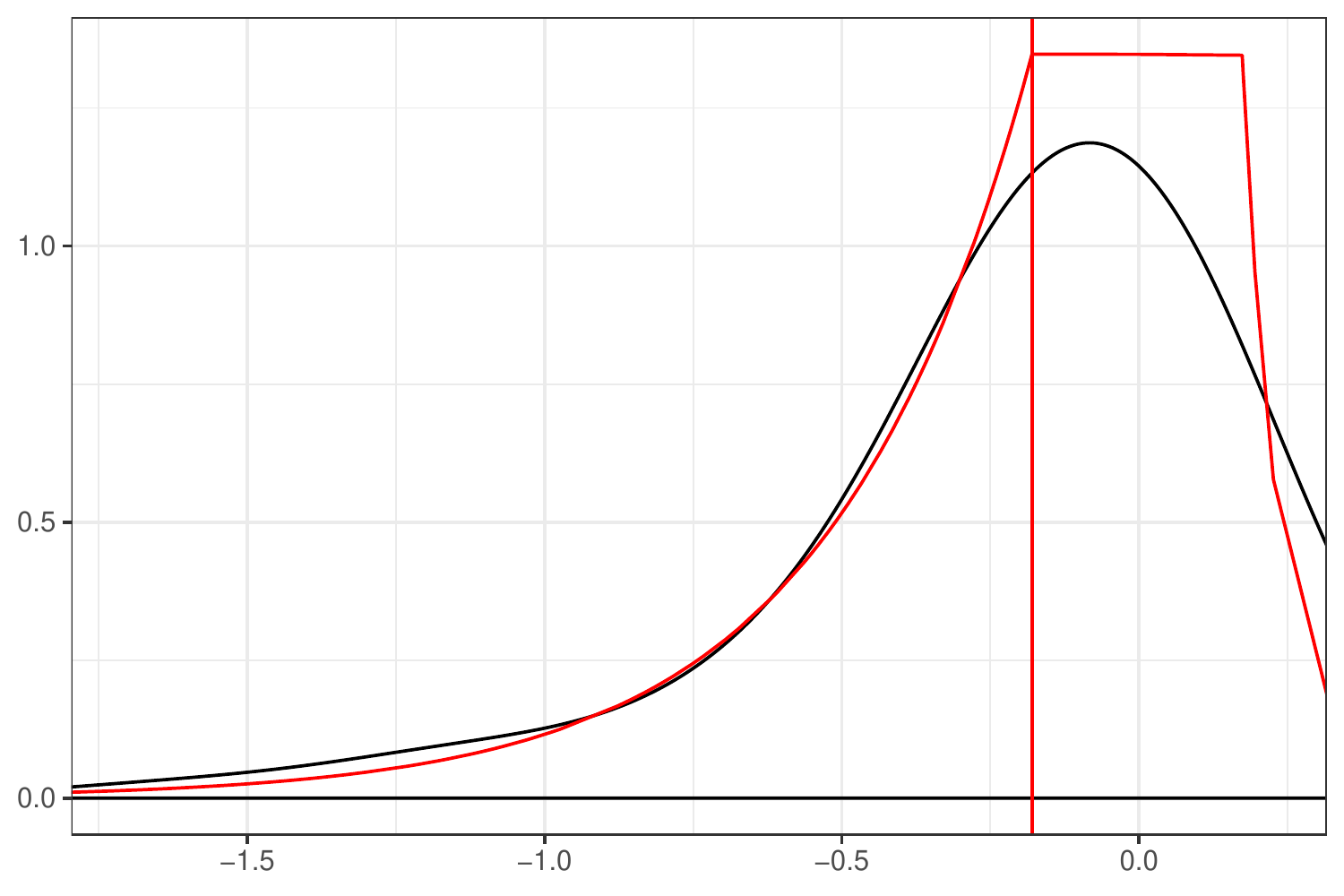}
    \caption{$\x = (-0.8859, 0.2587)$, $\v = (-0.9474, 0.32)$}
    \label{fig:rough-2}
  \end{subfigure}
  \caption{Shows the estimated projected weighted log-concave density estimator
    (red) and the kernel density estimator (black), for two slightly
    different starting points. The mode is indicated by the vertical line.}
  \label{fig:LCRS-jump}
\end{figure}

\begin{figure}[h]
  \begin{subfigure}[b]{0.5\textwidth}
    \includegraphics[width=\textwidth]{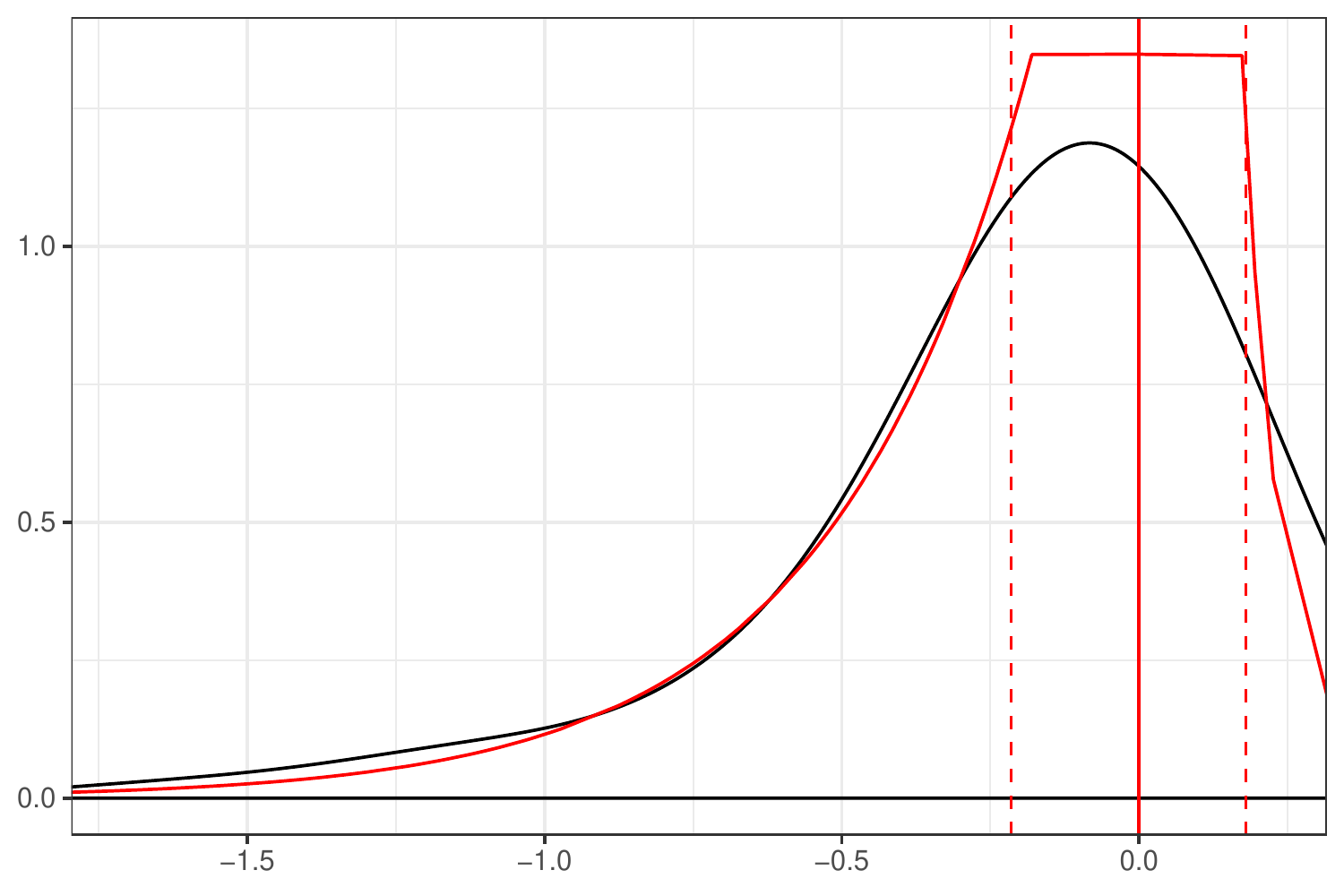}
    \caption{$\x = (-0.886, 0.2586)$, $\v = (-0.9475, 0.3199)$}
    \label{fig:rough-1-end}
  \end{subfigure}
  \hfill
  \begin{subfigure}[b]{0.5\textwidth}
    \includegraphics[width=\textwidth]{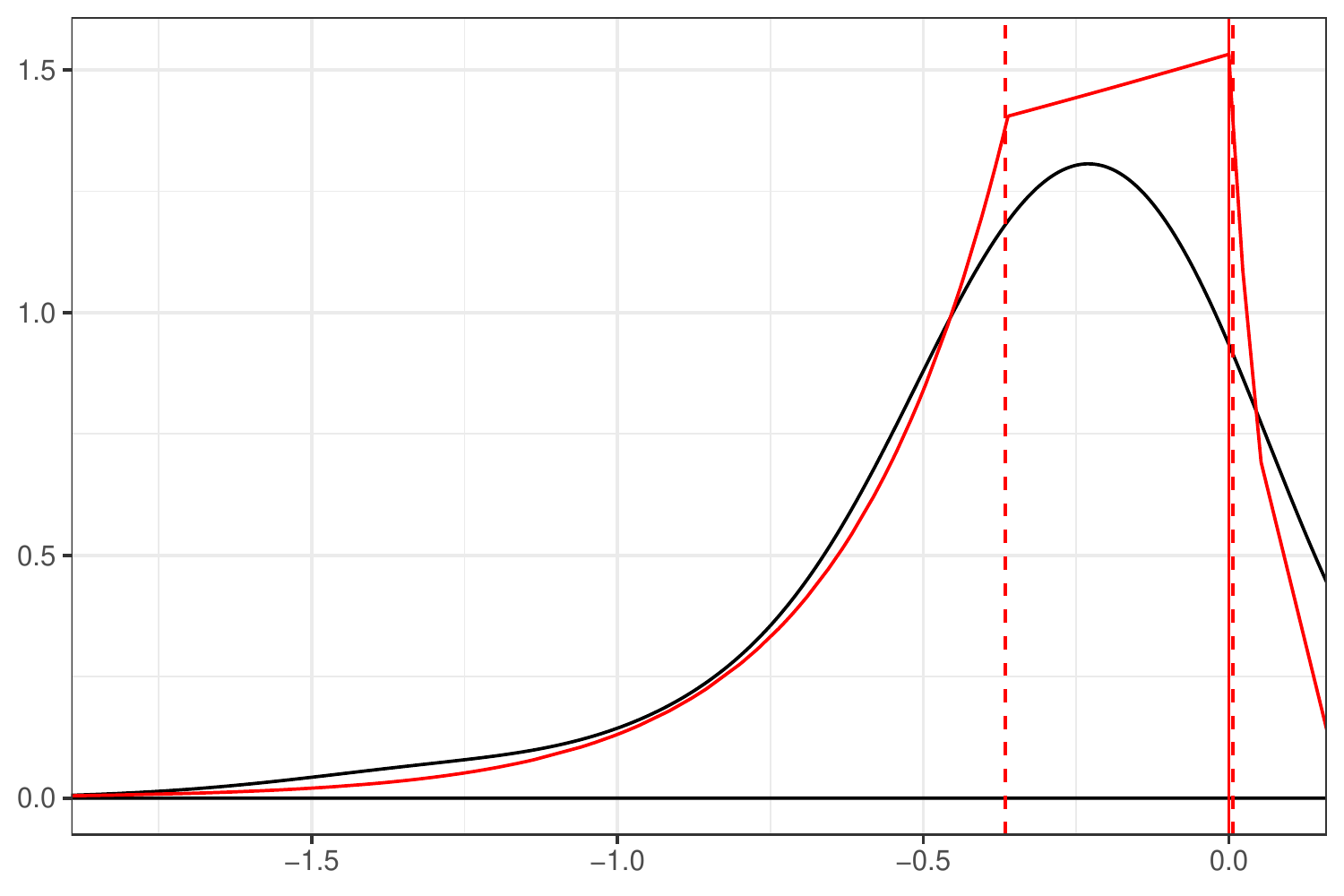}
    \caption{$\x = (-1.0467, 0.3241)$, $\v = (-0.9505, 0.3106)$}
    \label{fig:rough-2-end}
  \end{subfigure}
  \caption{Shows the estimated weighted projected log-concave density estimator
    (red) and the kernel density estimator (black) after the last step of
    the algorithm. The red dashed line corresponds to the threshold for
    $\alpha = 0.9$.}
  \label{fig:LCRS-jump-end}
\end{figure}

\begin{figure}
  \centering
  \includegraphics[width=.6\textwidth]{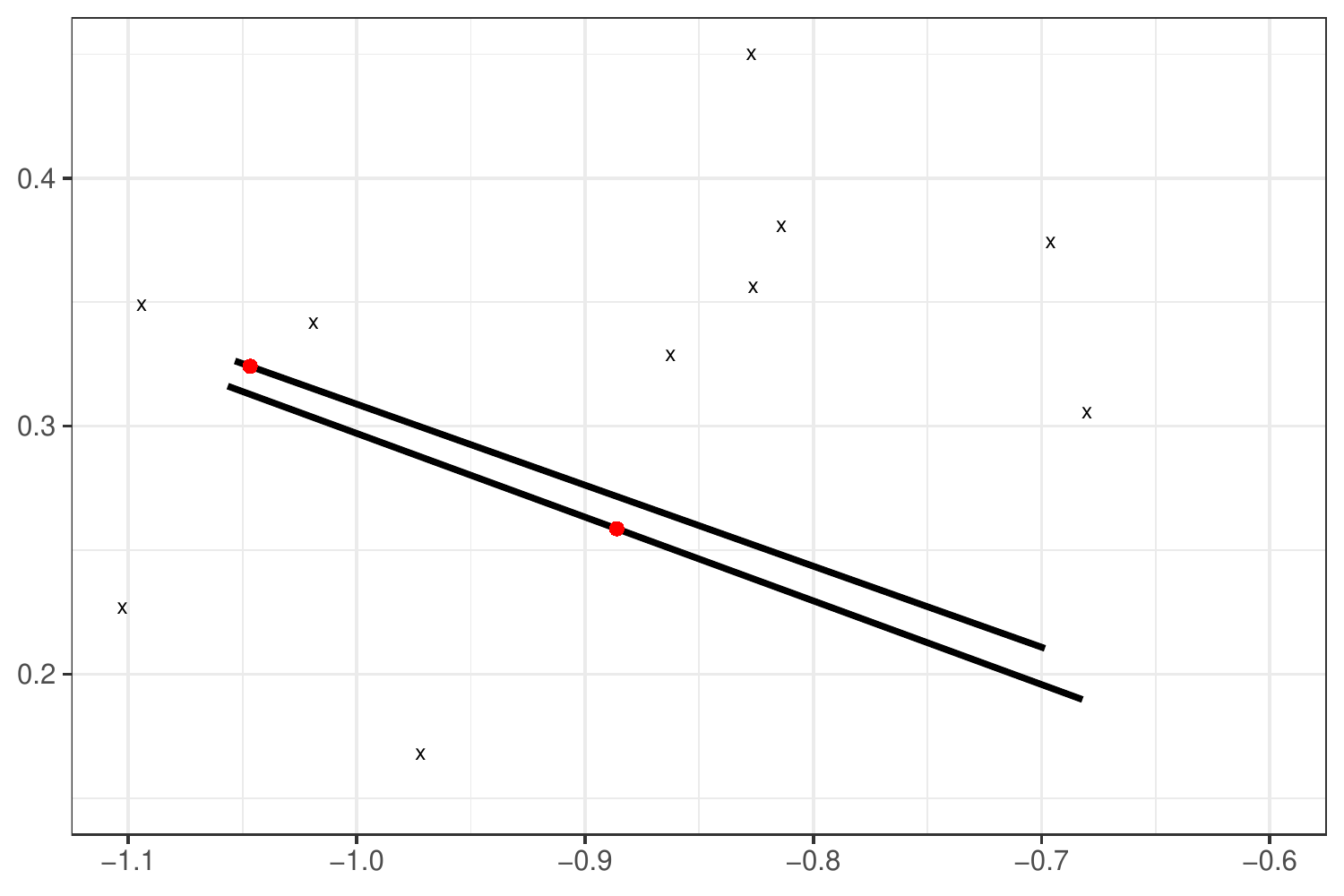}
  \caption{Shows the estimated ridge points (red) and the threshold
    interval (black) for two slightly different starting points.}
  \label{fig:jump-interval}
\end{figure}

\paragraph{Smoothed ridge.}
Another possibility, to avoid the discontinuous behavior of the estimated
ridge line, is by replacing the estimated log-density $\hat{\theta}$ by a
smoothed version. \citet{DuembgenRufibach2011} propose
\begin{displaymath}
  \hat{g}^*(y) = \int_{-\infty}^\infty \hat{g}(t) \phi_\gamma(y - t) \, dt,
\end{displaymath}
where $\hat{g}(t) = e^{\hat{\theta}(t)}$ and
$\phi_\gamma(t) \defeq (2\pi \gamma^2)^{-1/2} \exp(y^2 /
(2\gamma^2))$. This is the convolution of the estimated log-concave density
with a Gaussian distribution, where $\gamma$  is chosen, such that the
variance of the new estimator coincides with the variance of the empirical
distribution. The mode of $\hat{g}^*$ can be found by Newton's
method. Replacing the mode of $\hat{\theta}$ by the mode of $\hat{g}^*$ in
Algorithm 3 leads to typically smoother ridge lines. We will refer to it as
smoothed LCRS (sLCRS).

\section{Data Examples}
\label{sec:simulations}

The algorithms have been implemented in the statistical language R
\citet{R2019} and calculations were performed on UBELIX (http://www.id.unibe.ch/hpc), the HPC cluster at the University of Bern.

\subsection{Circle Data}
\label{sec:circle-data}

Suppose we have a sample $\X_1, \X_2, \ldots, \X_n \in \R^2$ with
distribution $\mathcal{P} \defeq \mathcal{L}(\X)$, where
\begin{displaymath}
  \X \defeq r
  \begin{pmatrix}
    \cos(2\pi U)\\
    \sin(2\pi U)\\
  \end{pmatrix}
  + \sigma \Z,
\end{displaymath}
with $r, \sigma \in \R_{>0}$ and independent random variables $U \sim
\mathcal{U}([0,1])$ and $\Z \sim \NN(0, \I_2)$. The density function $f$ of
$\mathcal{P}$ is
\begin{displaymath}
  f(\x) = (2\pi \sigma^2)^{-1} I_o\bigl( r / \sigma^2 \lVert \x \rVert
    \bigr) \exp\Bigl( - \frac{r^2 + \lVert \x \rVert^2}{2 \sigma^2}
    \Bigr),
\end{displaymath}
where
\begin{displaymath}
  I_0(t) \defeq \sum_{m=0}^\infty \frac{(t^2 / 4)^m}{(m!)^2}
\end{displaymath}
is the modified Bessel function of the first kind with parameter $0$. One
can show that the ridge of $f$ is the origin if $r / \sigma \le \sqrt{2}$
and a circle with radius in $(0,r)$ and centre at the origin if
$r / \sigma > \sqrt{2}$. The exact value can be numerically calculated
with bisection; see Section \ref{sec:ridge-circle-data} for details. In our
simulation study we use $n = 200$ data points with $r = 1$ and $\sigma =
0.1$, the ridge is then the circle with center $\bs{0}$ and radius
$0.995$. We estimate the ridge with the SCMS, LCRS and sLCRS
algorithm based on the Gaussian kernel $K_h$ for different bandwidths; see
Figure \ref{fig:circle-compare}.

\paragraph{Results.} 

The ridge estimated by SCMS is biased towards the center and the larger the
bandwidth $h$, the larger the biases. Indeed, let $\Z_h$ have density
function $K_h$, then we estimate the Ridge of the random variable
$\X + \Z_h$ instead of $\X$, which has the same distribution as circle data
with standard deviation $\sqrt{\sigma^2 + h^2}$ instead of $\sigma$. 

The ridge estimated by LCRS is not sensitive to the choice of the
bandwidth. However, the estimated ridge is discontinuous and we should use
confidence or threshold intervals to show this uncertainty; see Figure
\ref{fig:LCRS-interval}. There we see, that e.g.~on the top left, the
estimator is very uncertain. A close look at the data reveals indeed, that
the data point in this area are spread away from the true ridge, whence the
ridge could lay in a wide region. On the top right, the data points are
nicely spread around the true ridge and the LCRS catches that well.

For the sLCRS we still get a discontinuous ridge for $h = 0.2$. However,
for larger bandwidths it is smooth and the bias caused by the smoothing is
less serious then for the SCMS.

\begin{figure}
  \begin{subfigure}[b]{\textwidth}
    \caption{SCMS}
    \includegraphics[width=\textwidth]{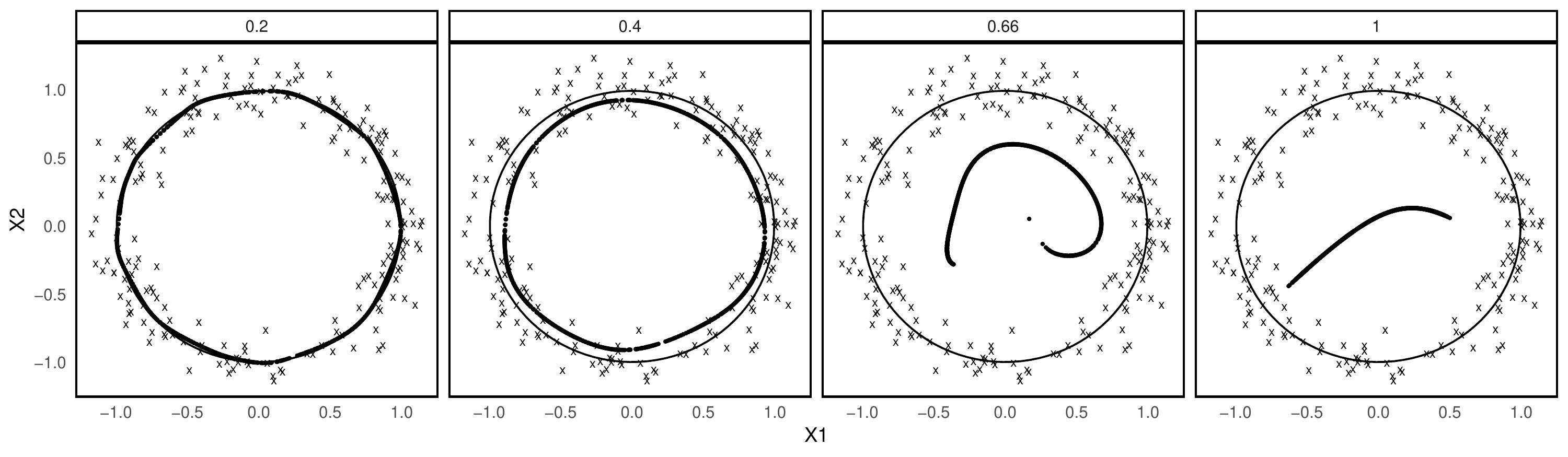}
    \label{fig:SCMS}
  \end{subfigure}
  \hfill
  \begin{subfigure}[b]{1\textwidth}
    \caption{LCRS}
    \includegraphics[width=\textwidth]{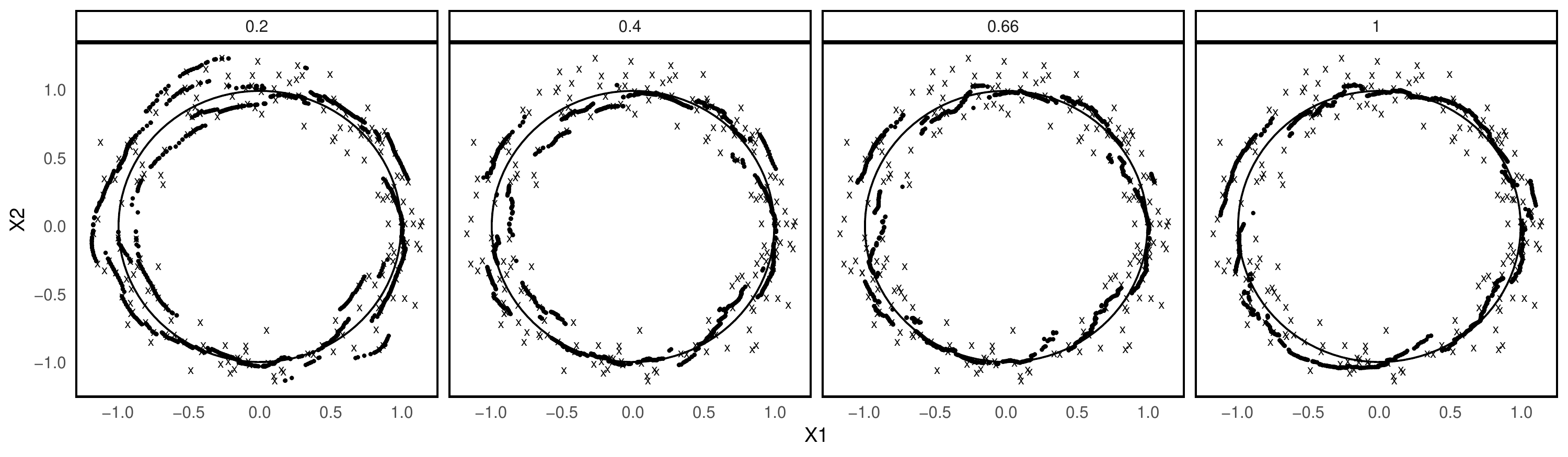}
    \label{fig:LCRS}
  \end{subfigure}
  \hfill
    \begin{subfigure}[b]{1\textwidth}
    \caption{sLCRS}
    \includegraphics[width=\textwidth]{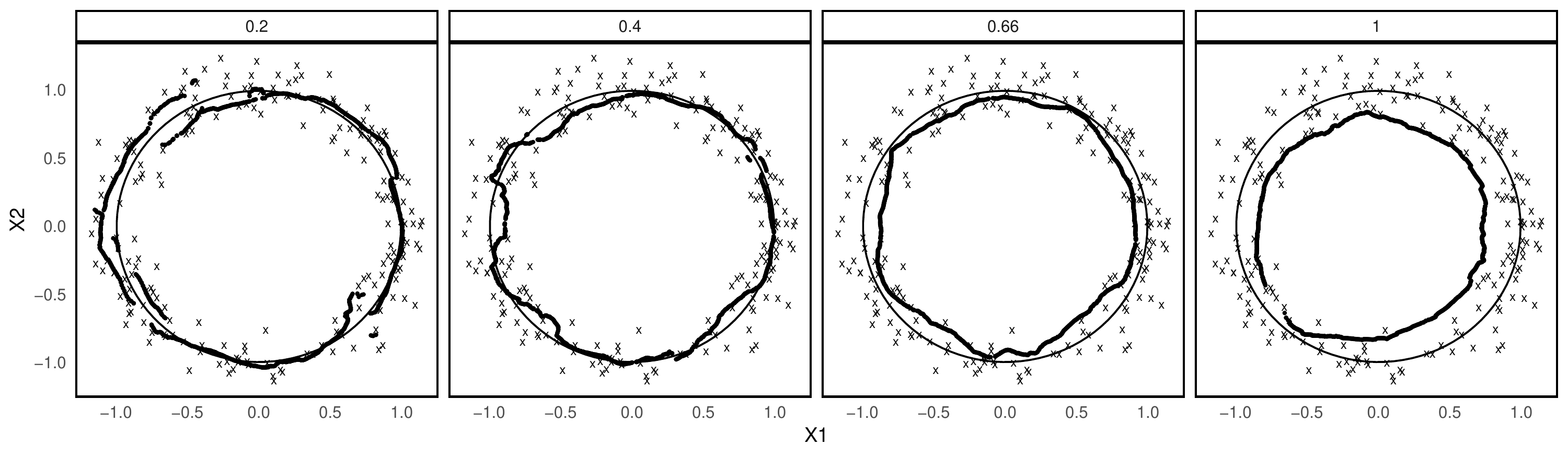}
    \label{fig:sLCRS}
  \end{subfigure}
  \caption{Shows the estimated ridge points for different bandwidths for
    both algorithms and the true ridge as a solid circle.}
  \label{fig:circle-compare}
\end{figure}

\begin{figure}
  \begin{subfigure}[b]{0.5\textwidth}
    \centering
    \caption{Threshold interval}
    \includegraphics[width=\textwidth]{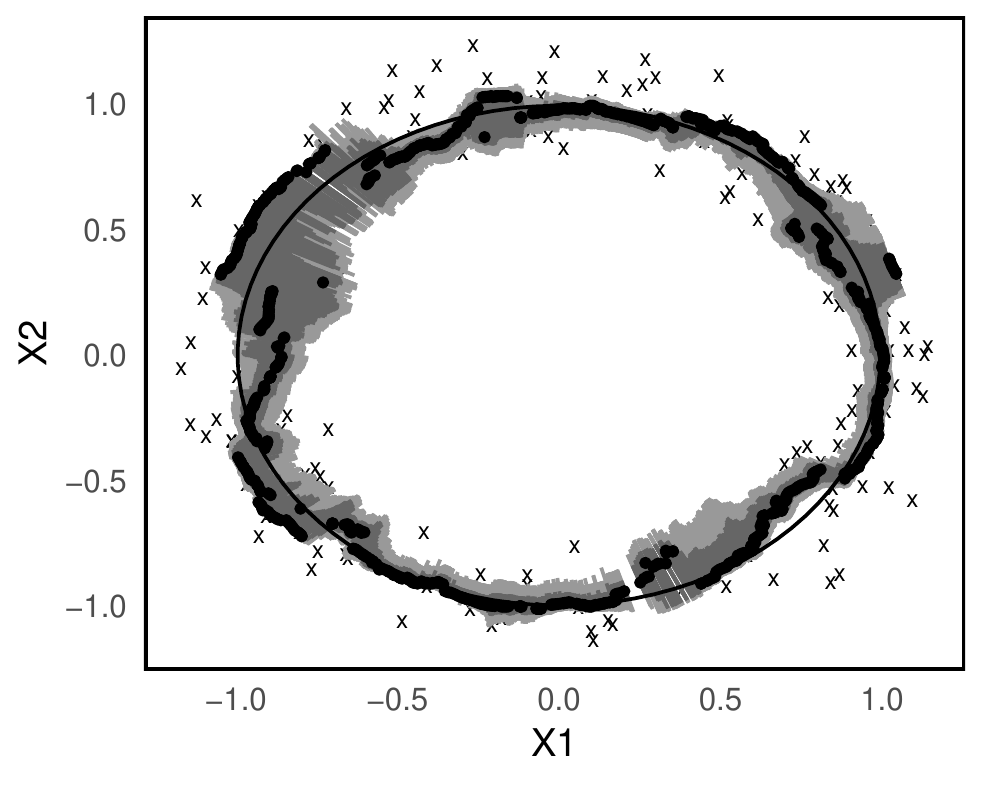}
    \label{fig:circle-interval}
  \end{subfigure}
  \begin{subfigure}[b]{0.5\textwidth}
    \caption{Confidence interval}
    \includegraphics[width=\textwidth]{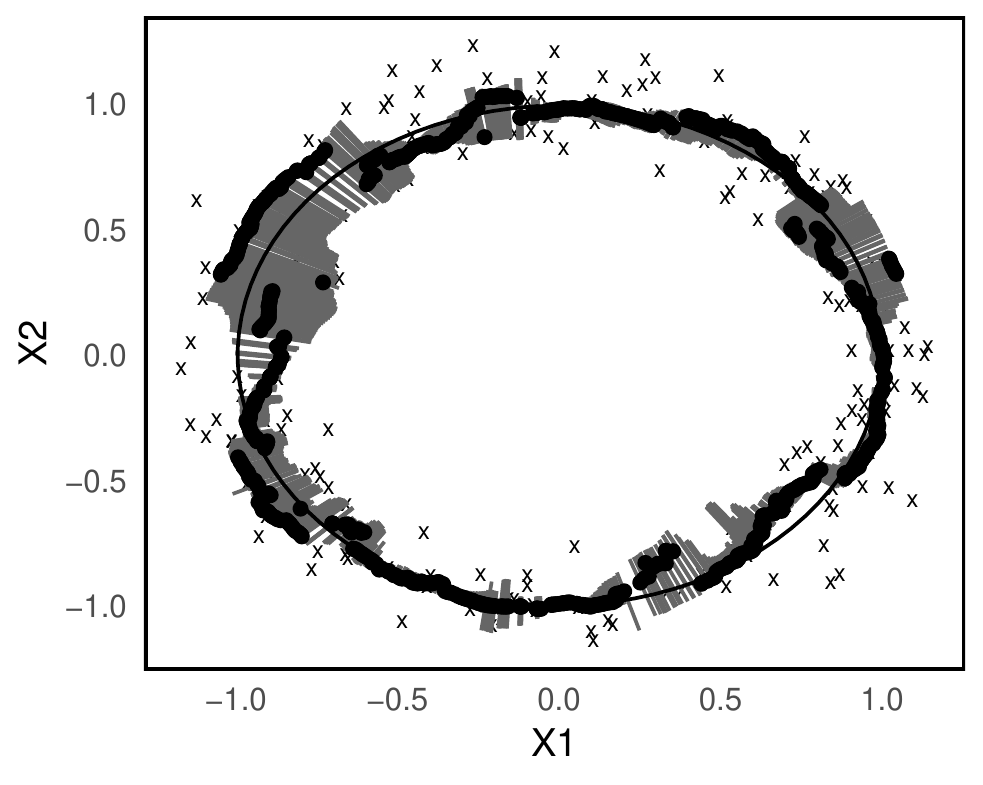}
    \label{fig:circel-conf-int}
  \end{subfigure}
  \caption{Shows the estimated ridge points and true ridge (black). Left:
    Threshold interval for $\alpha \in \{0.8, 0.9\}$. Right: 90\%-Confidence
    interval for the ridge point in direction of smallest variance.}
  \label{fig:LCRS-interval}
\end{figure}

\subsection{Galaxy filaments}
\label{sec:galaxy}

\begin{figure}
  \begin{subfigure}[b]{0.5\textwidth}
    \centering
    \caption{Low redshift}
    \includegraphics[width=\textwidth]{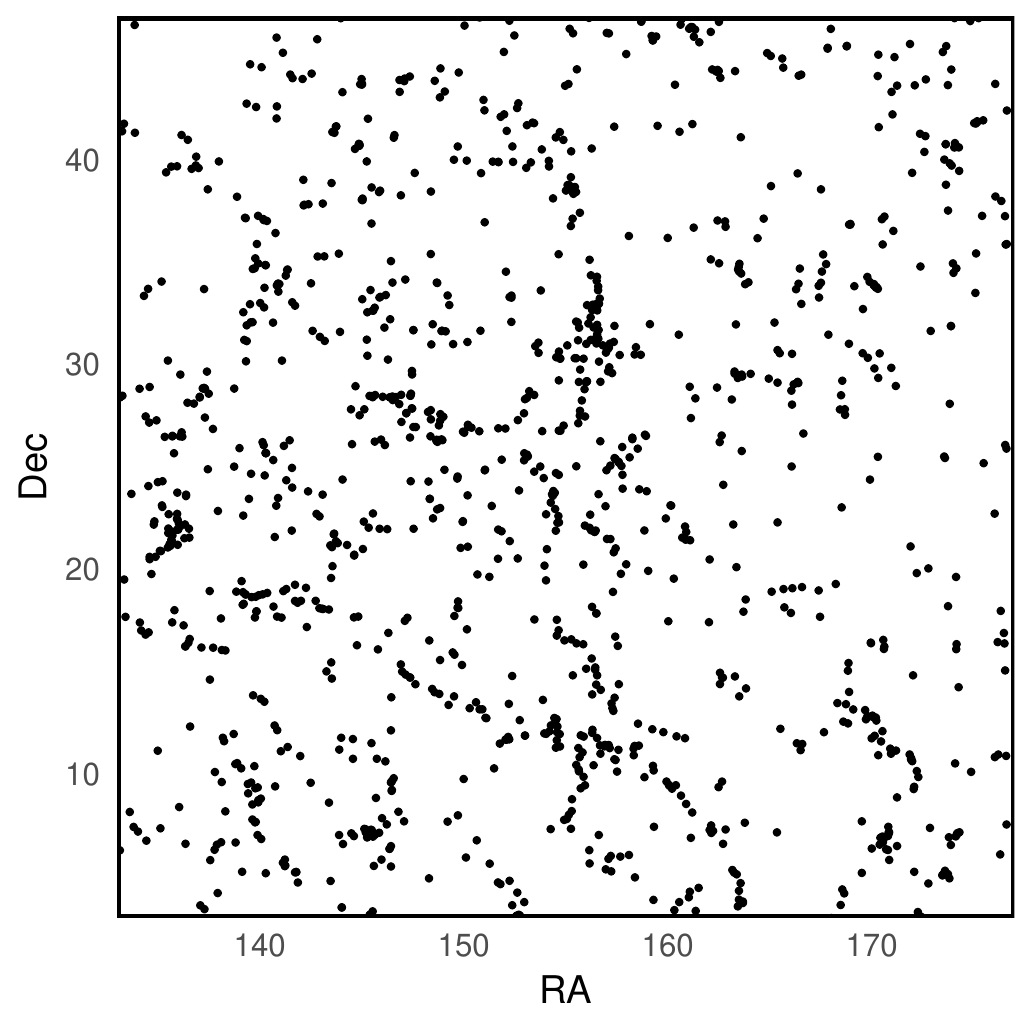}
    \label{fig:data-low}
  \end{subfigure}
  \begin{subfigure}[b]{0.5\textwidth}
    \caption{High redshift}
    \includegraphics[width=\textwidth]{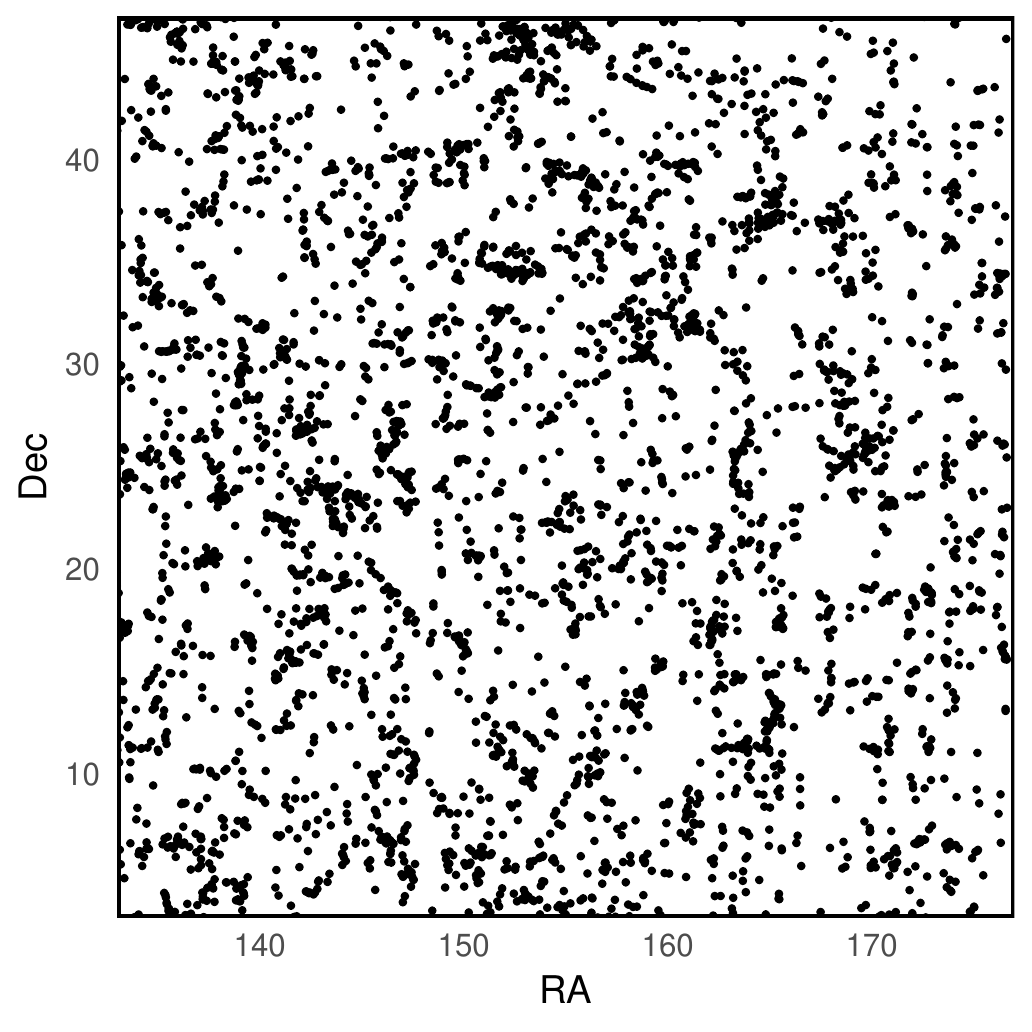}
    \label{fig:data-high}
  \end{subfigure}
  \caption{Galaxy structure from the Baryon Oscilation Spectorspic Survey
    for a slice of sky at two different redshifts.}
  \label{fig:data}
\end{figure}

\begin{figure}[t]
  \begin{subfigure}[b]{0.33\textwidth}
    \centering
    \caption{SCMS for low redshift}
    \includegraphics[width=\textwidth]{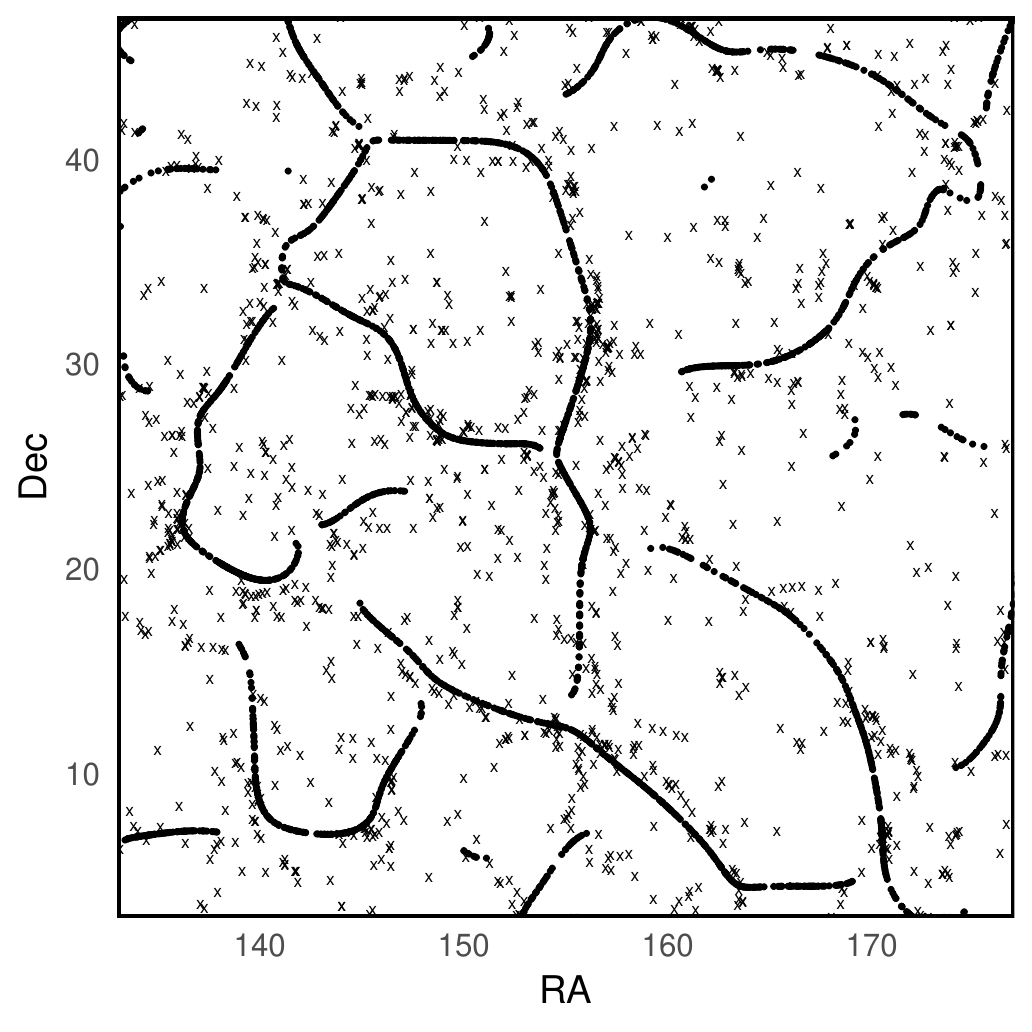}
    \label{fig:cosmic-web-SCMS-Fig12}
  \end{subfigure}
  \begin{subfigure}[b]{0.33\textwidth}
    \caption{LCRS for low redshift}
    \includegraphics[width=\textwidth]{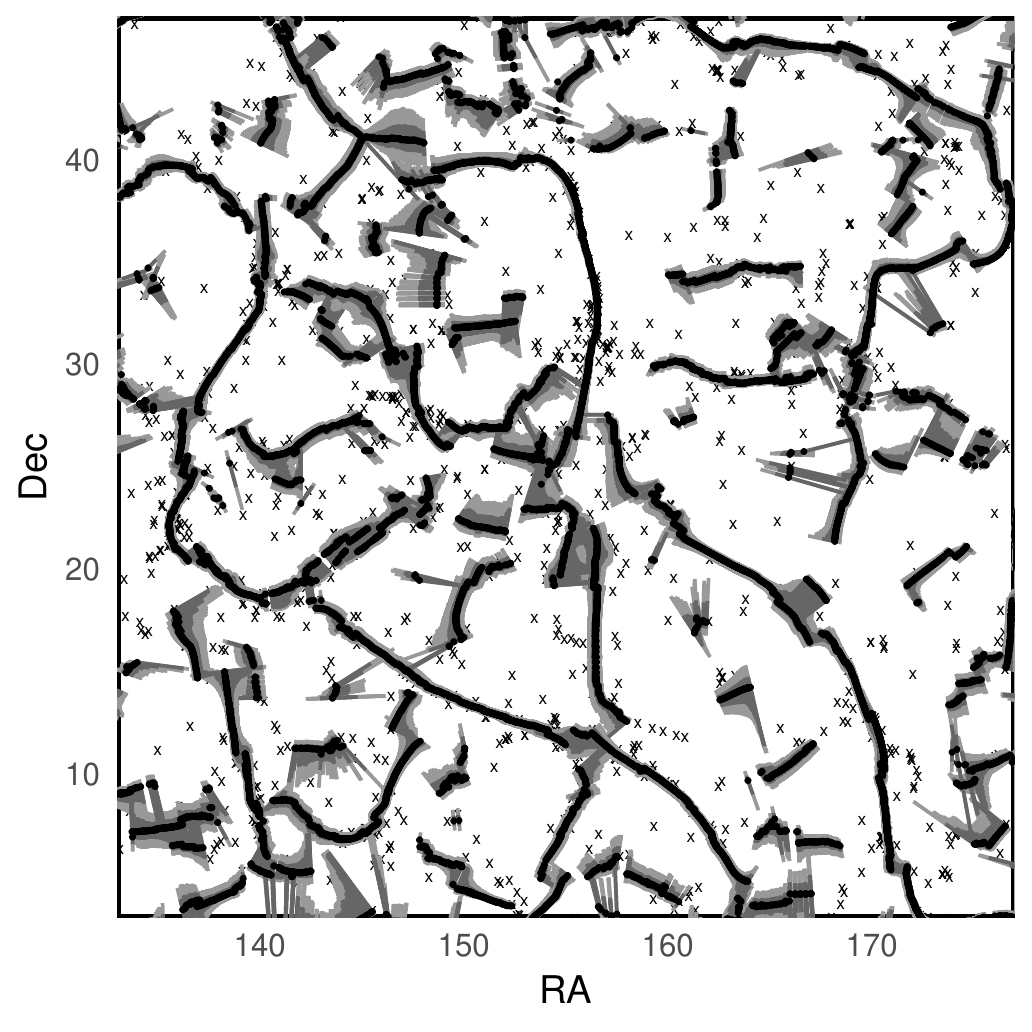}
    \label{fig:cosmic-web-LCRS-Fig12}
  \end{subfigure}
    \begin{subfigure}[b]{0.33\textwidth}
    \caption{sLCRS for low redshift}
    \includegraphics[width=\textwidth]{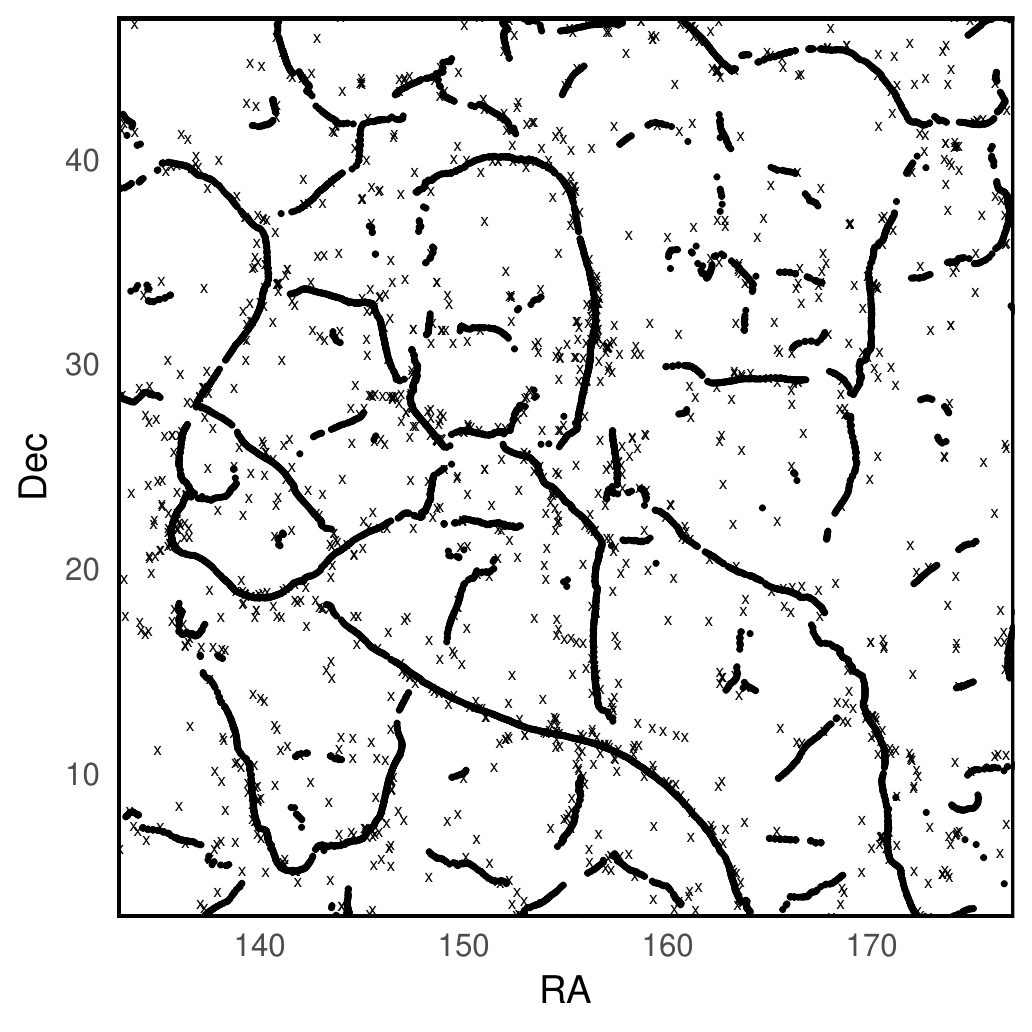}
    \label{fig:cosmic-web-sLCRS-Fig12}
  \end{subfigure}
  \begin{subfigure}[b]{0.33\textwidth}
    \caption{SCMS for high redshift}
    \includegraphics[width=\textwidth]{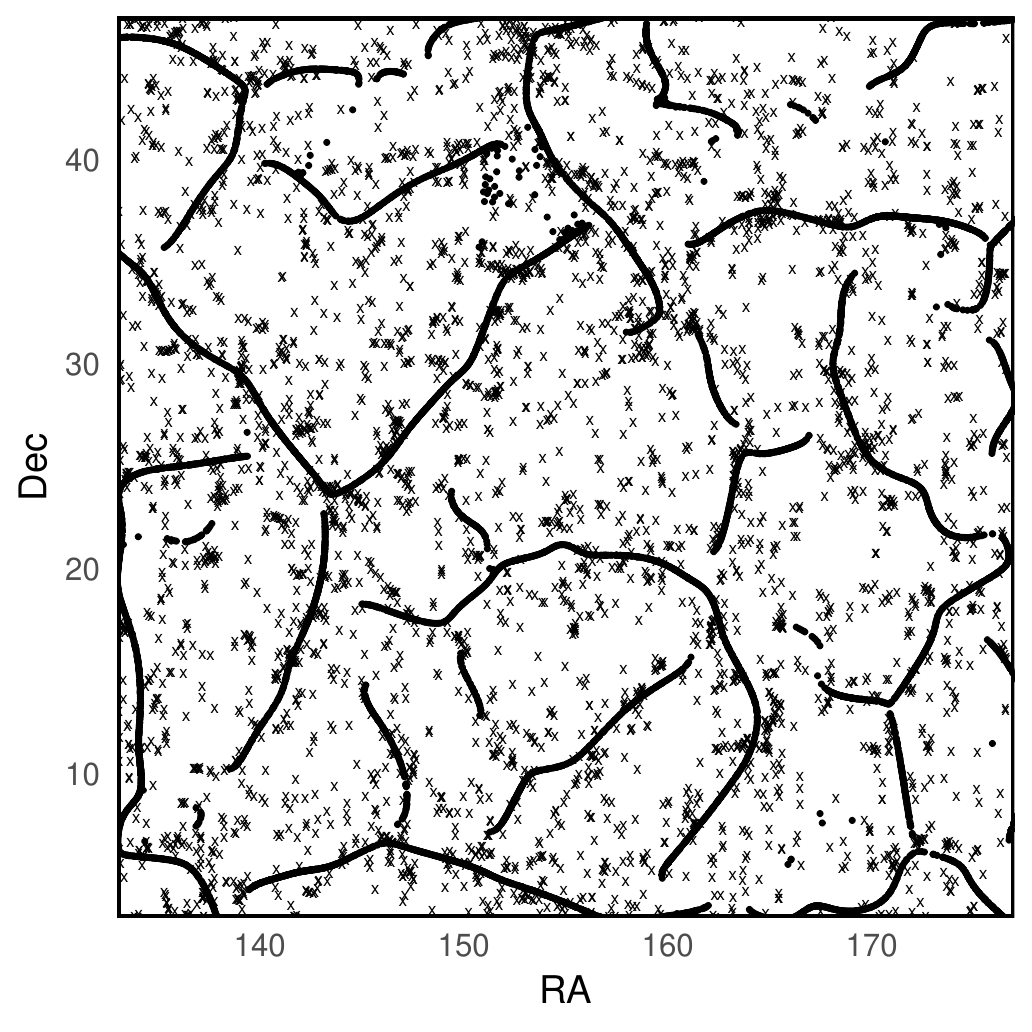}
    \label{fig:cosmic-web-SCMS-Fig13}
  \end{subfigure}
  \begin{subfigure}[b]{0.33\textwidth}
    \caption{LCRS for high redshift}
    \includegraphics[width=\textwidth]{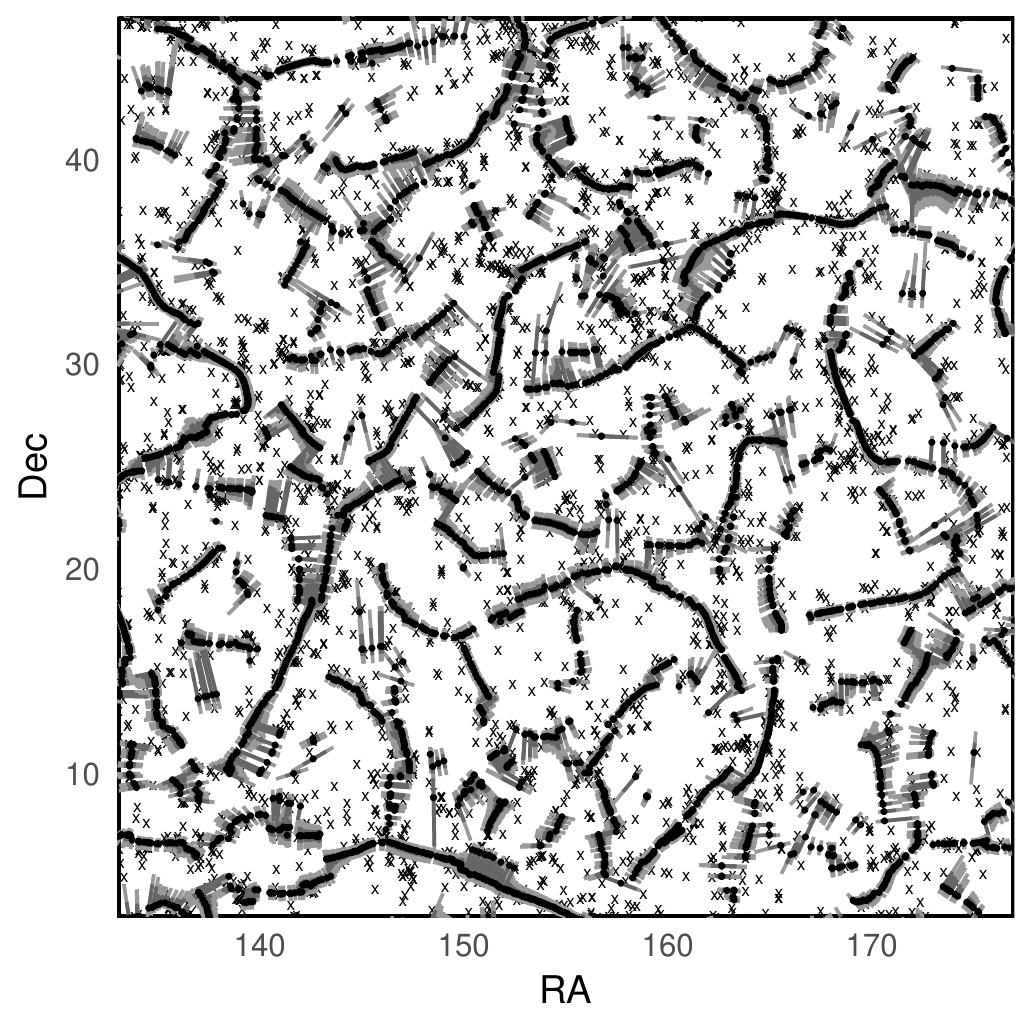}
    \label{fig:cosmic-web-LCRS-Fig13}
  \end{subfigure}
    \begin{subfigure}[b]{0.33\textwidth}
    \caption{sLCRS for high redshift}
    \includegraphics[width=\textwidth]{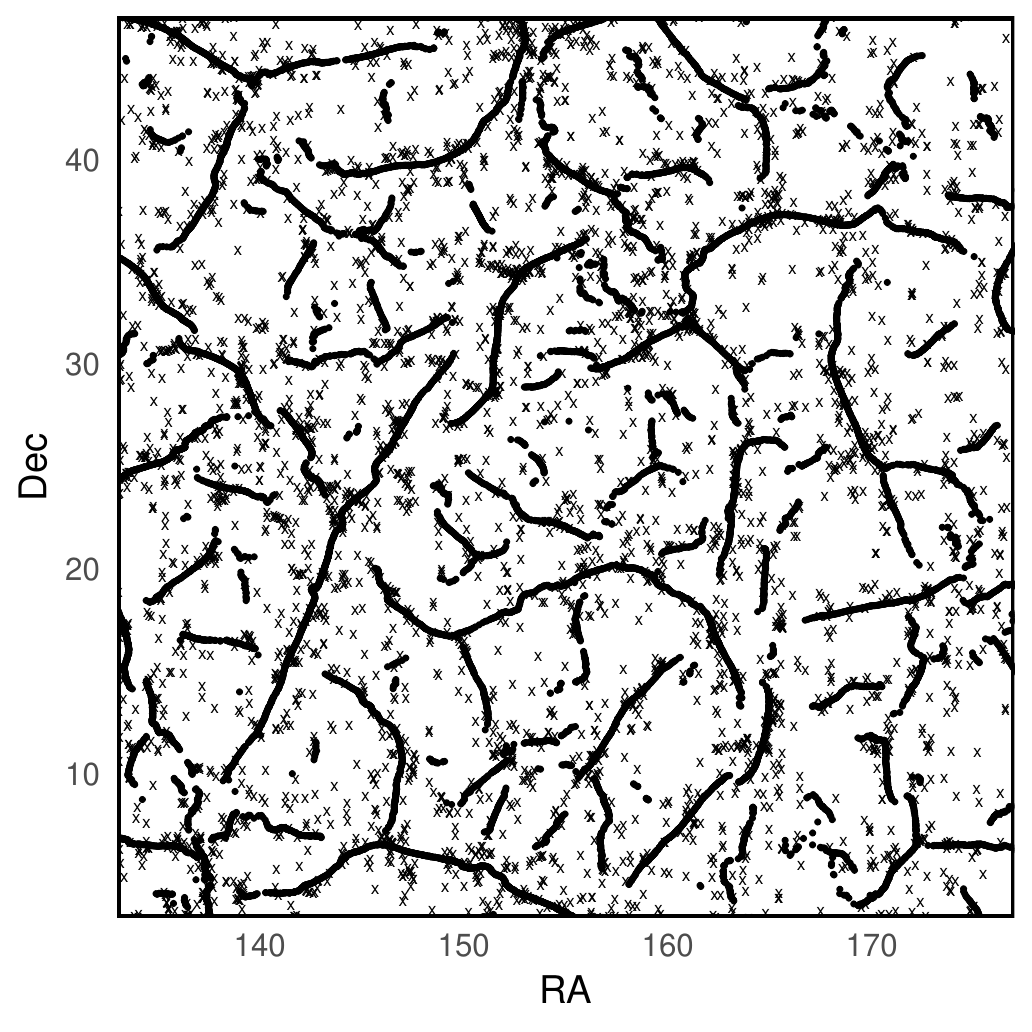}
    \label{fig:cosmic-web-sLCRS-Fig13}
  \end{subfigure}
  \caption{Shows the estimated ridge points (black) for low and high
    redshift. Threshold interval (grey) with
    $\alpha \in \{0.8,0.9\}$ are drawn for the LCRS. The bandwidth is $h =
    2.03$.}
  \label{fig:cosmic-web}
\end{figure}

We apply both
algorithms to data from Data Release 16 (\citet{AhumadaETAL2019}) of the
Sloan Digital Sky Survey (SDSS); see \citet{YorkETAL2000} and \citet{EisensteinETAL2011}.

The Baryon Oscilation Spectroscopic Survey (BOSS) is part of the SDSS and
obtains the redshift $z$ from 1.5 million luminous galaxies on 10 000
square degrees of sky in celestial coordinates. The longitude is called
right ascension (RA) and the latitude is called declination (Dec). Both are
measured in degree.

Th galaxies in our universe are not distributed uniformly, they follow a
web structure with clusters, sheets and empty voids. The filaments are
one-dimensional structures connecting clusters and build the boundaries of
the voids. The knowledge about filaments at the range of different
redshifts is interesting for cosmologists to study the evolution of the
universe. For more detailed information we refer to \citet{Chen2015} and
the reference therein.

\citet{Chen2015} used Data Release 12 to estimate galaxy filaments at
different redshifts with the
SCMS algorithm. We use
the same slices of data, namely at
\begin{displaymath}
  135^\circ \le \text{ RA } \le 175^\circ \quad 5^\circ \le \text{Dec} \le
  45^\circ \quad 0.245 \le z \le 0.240 
\end{displaymath}
with low redshift and at
\begin{displaymath}
  135^\circ \le \text{ RA } \le 175^\circ \quad 5^\circ \le \text{Dec} \le
  45^\circ \quad 0.530 \le z \le 0.535
\end{displaymath}
with high redshift. We apply both algorithms to the data for low and high
redshift and compare them to each other. The used slice for the algorithms
are from $130^\circ \leq \text{ RA } \leq 180^\circ$ and $0^\circ \leq
\text{ Dec } \leq 50^\circ$ to avoid boundary effects in the estimation. We
refrain from analysing the data in 3 dimensions, where redshift could be
used as the third one, for different reasons. The obvious one being, that
with the LCRS algorithm one can only estimate 1-dimensional ridges in 2
dimensions, but not in 3. Another reason being, that the density of
galaxies changes for different redshift, whence, estimating everything with
the same bandwidth may be problematic.

We focus on the bandwidth given in \eqref{eq:Silverman} with $A_0 =
0.4$. This particular choice of $A_0$ was made in \citet{Chen2015} by
trying different ones on taking the most suitable one, leading to $h = 2.5$ in
the low redshift and $h = 2.03$ in the high redshift data.

For the high redshift data, we show the results for different bandwidths
for the SCMS and LCRS algorithm; see Figure \ref{fig:SCMS-bandwidths} and
\ref{fig:LCRS-bandwidths}. The optimal bandwidth calculated based on the
euclidean minimal spanning tree is $h = 0.84$. Additionally, we considered
$h = 1.5$ and $h = 2.5$.

As starting points we choose a grid of points with vertical and horizontal
distance $0.5$ and remove all point further then $0.5$ away from any data
point. For the low redshift we get $2842$ points and for the high redshift
data $5777$ points. For sLCRS we used a finer grid, where the points lay
$0.25$ apart, leading to $23080$ grid points. \citet{Chen2015} chose the
data points as staring points and removed those, where the estimated
density was below a certain threshold (low redshift: $1.02 \cdot 10^{-3}$,
high redshift: $7.52 \cdot 10^{-4}$).

\begin{figure}
  \begin{subfigure}[b]{0.24\textwidth}
    \centering
    \caption{$h = 0.84$}
    \includegraphics[width=\textwidth]{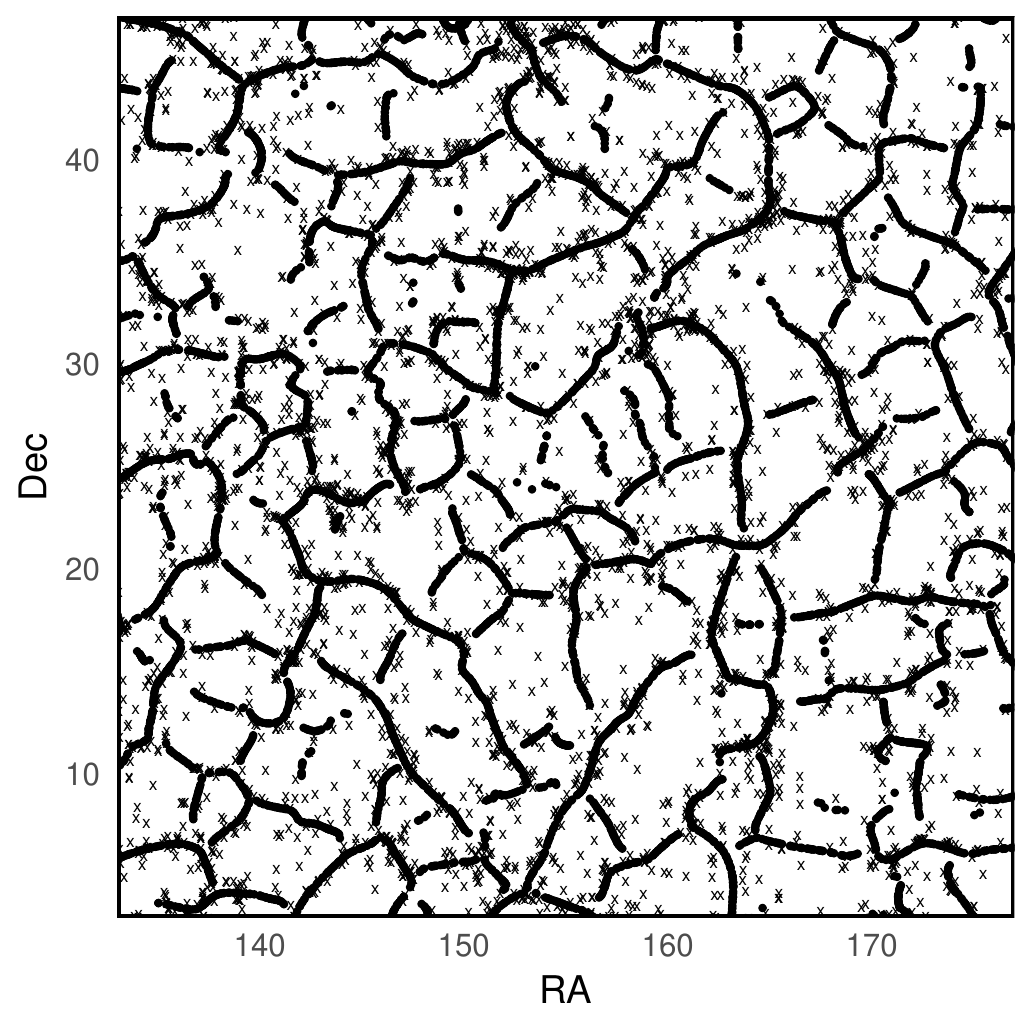}
    \label{fig:cosmic-web-LCRS-Fig12-084}
  \end{subfigure}
  \begin{subfigure}[b]{0.24\textwidth}
    \caption{$h = 1.5$}
    \includegraphics[width=\textwidth]{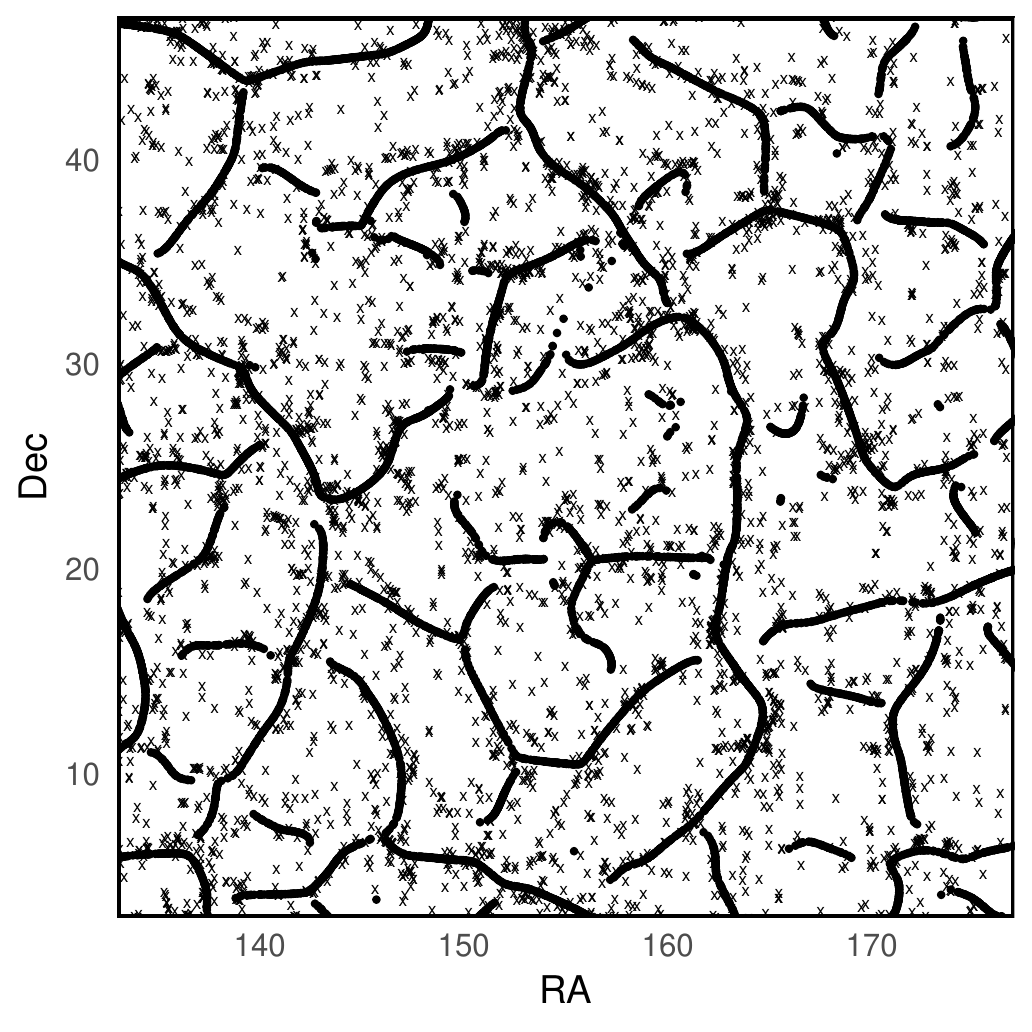}
    \label{fig:cosmic-web-LCRS-Fig12-150}
  \end{subfigure}
  \begin{subfigure}[b]{0.24\textwidth}
    \caption{$h = 2.03$}
    \includegraphics[width=\textwidth]{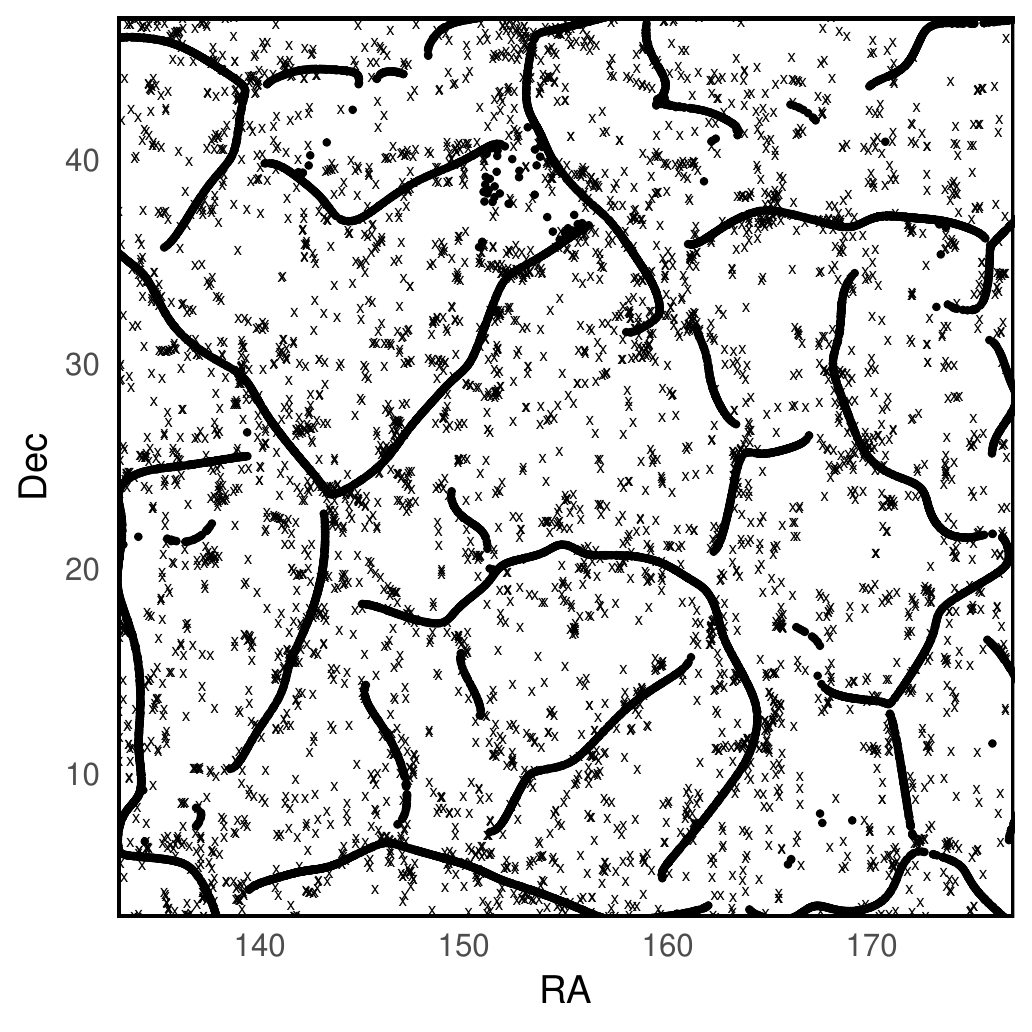}
    \label{fig:cosmic-web-SCMS-Fig13-203}
  \end{subfigure}
  \begin{subfigure}[b]{0.24\textwidth}
    \caption{$h = 2.5$}
    \includegraphics[width=\textwidth]{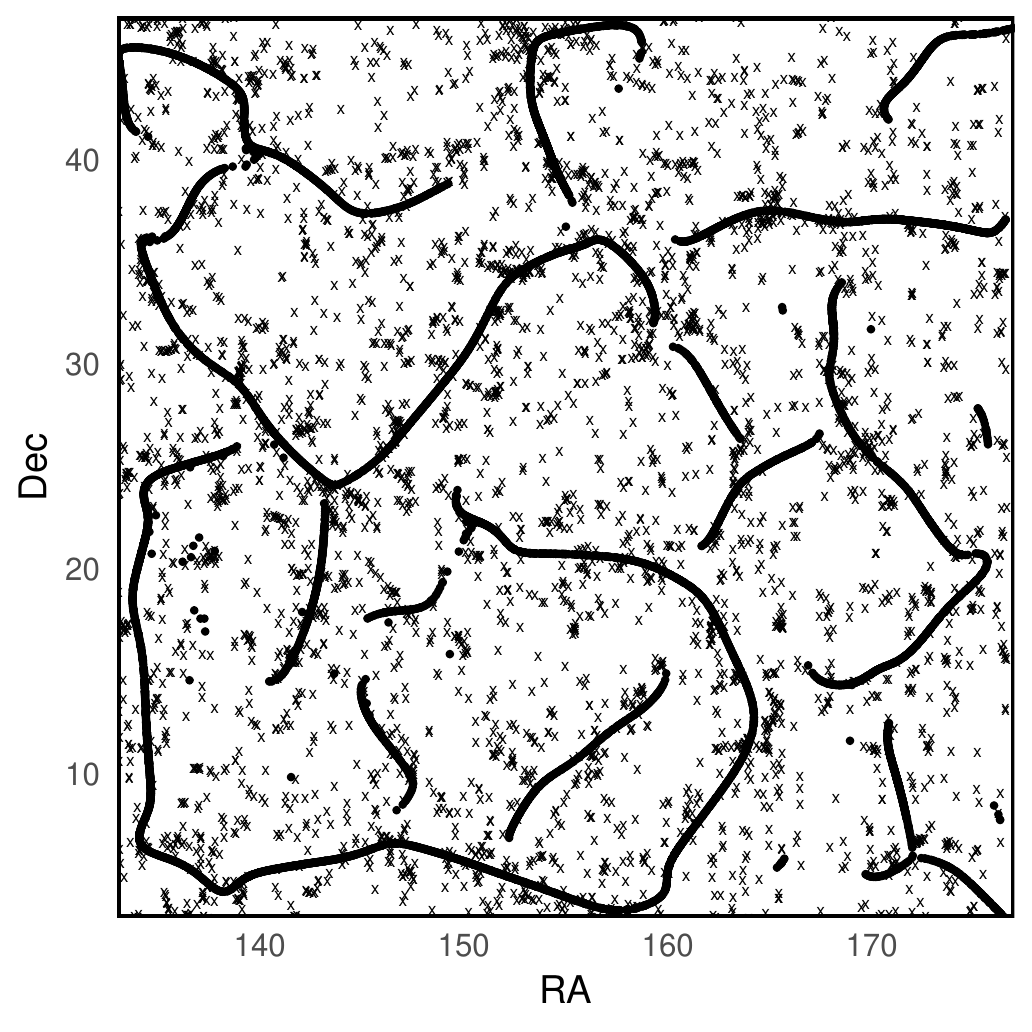}
    \label{fig:cosmic-web-LCRS-Fig13-250}
  \end{subfigure}
  \caption{Shows the estimated ridge points for the SCMS with different
    bandwidths $h$.}
  \label{fig:SCMS-bandwidths}
\end{figure}

\begin{figure}
  \begin{subfigure}[b]{0.24\textwidth}
    \centering
    \caption{$h = 0.84$}
    \includegraphics[width=\textwidth]{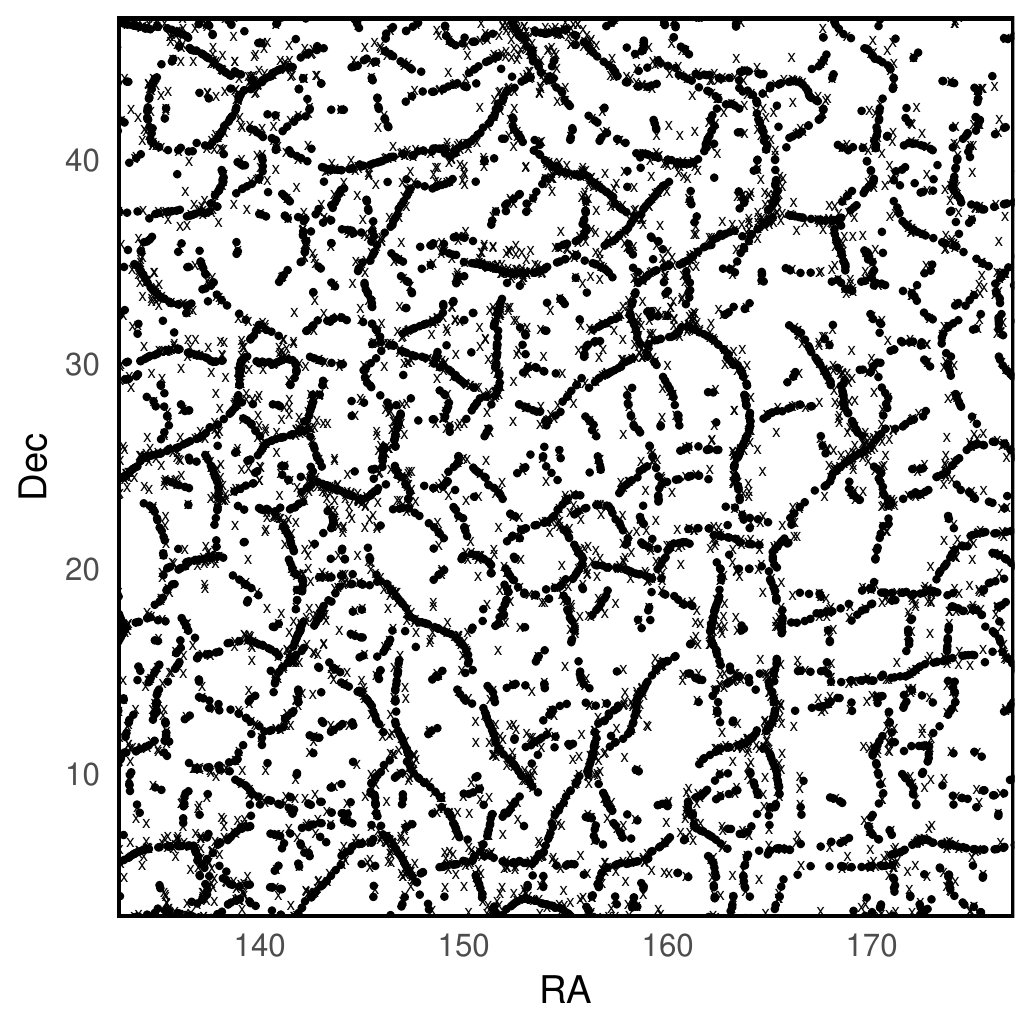}
    \label{fig:cosmic-web-LCRS-Fig13-084}
  \end{subfigure}
  \begin{subfigure}[b]{0.24\textwidth}
    \caption{$h = 1.5$}
    \includegraphics[width=\textwidth]{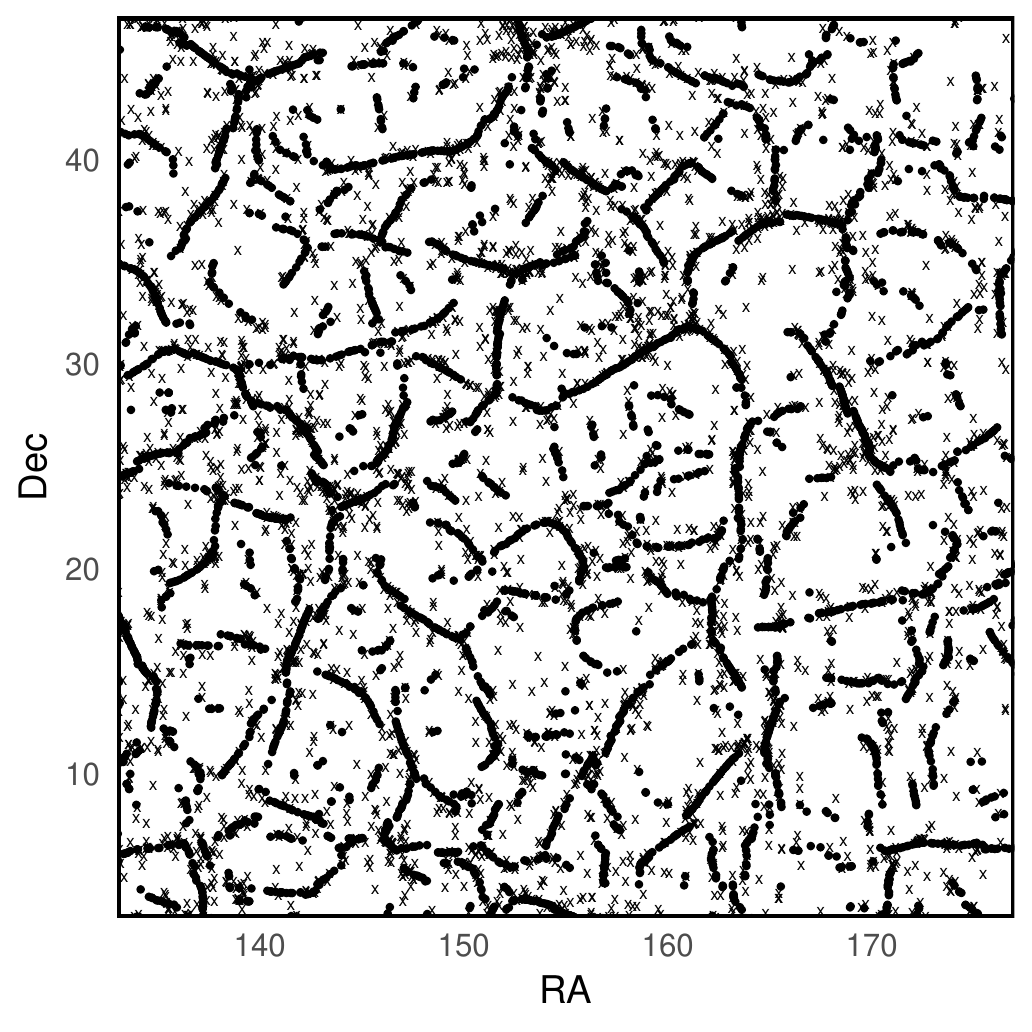}
    \label{fig:cosmic-web-LCRS-Fig13-150}
  \end{subfigure}
  \begin{subfigure}[b]{0.24\textwidth}
    \caption{$h = 2.03$}
    \includegraphics[width=\textwidth]{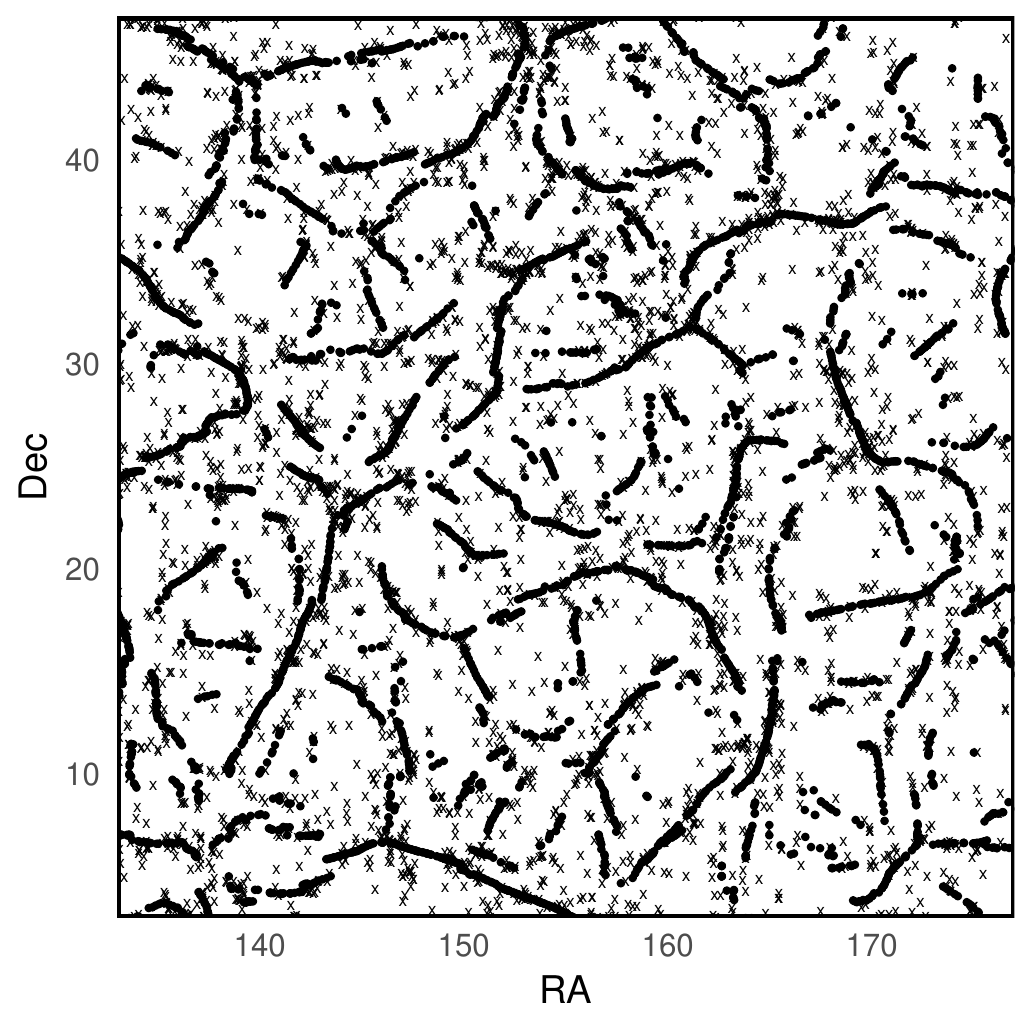}
    \label{fig:cosmic-web-SCMS-Fig13-203}
  \end{subfigure}
  \begin{subfigure}[b]{0.24\textwidth}
    \caption{$h = 2.5$}
    \includegraphics[width=\textwidth]{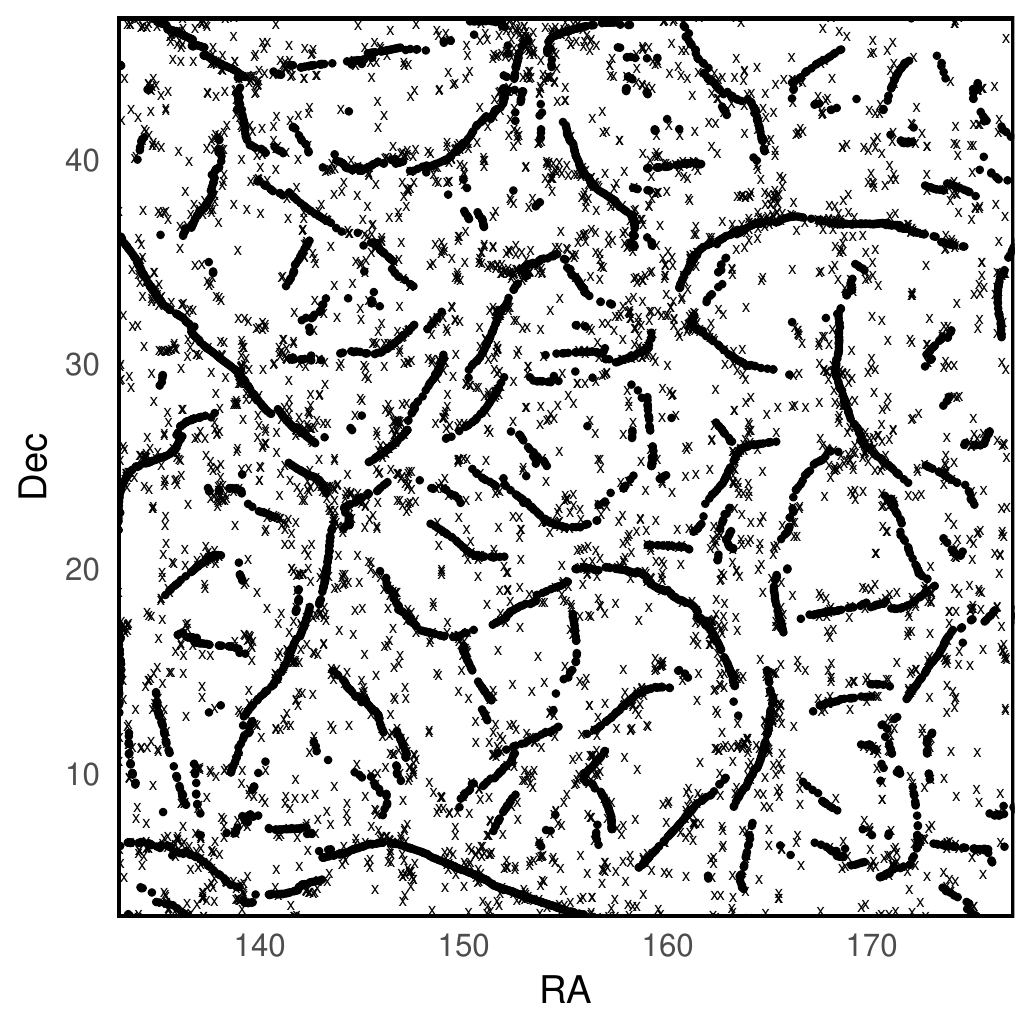}
    \label{fig:cosmic-web-LCRS-Fig13-250}
  \end{subfigure}
  \caption{Shows the estimated ridge points for the LCRS with different
    bandwidths $h$.}
  \label{fig:LCRS-bandwidths}
\end{figure}

\begin{figure}
  \begin{subfigure}[b]{0.24\textwidth}
    \centering
    \caption{$h = 0.84$}
    \includegraphics[width=\textwidth]{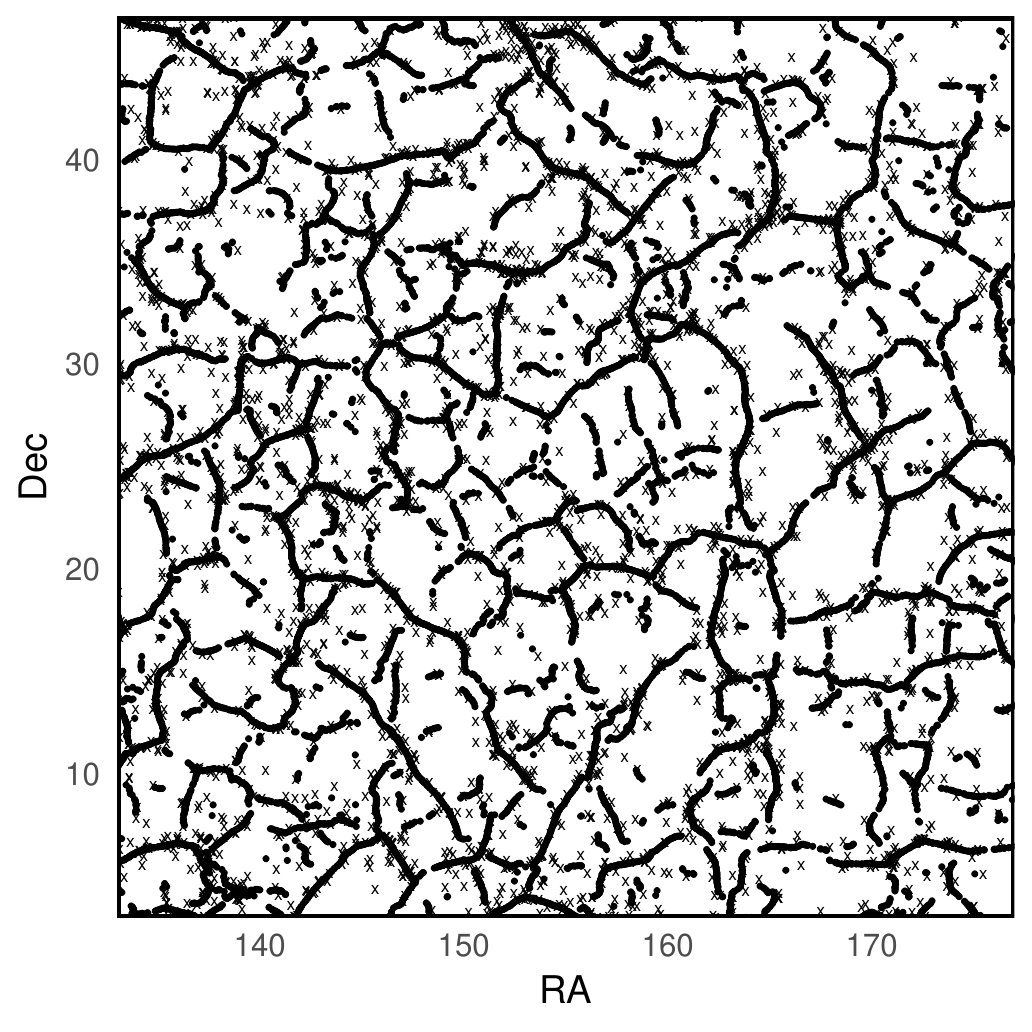}
    \label{fig:cosmic-web-LCRS-Fig13-084}
  \end{subfigure}
  \begin{subfigure}[b]{0.24\textwidth}
    \caption{$h = 1.5$}
    \includegraphics[width=\textwidth]{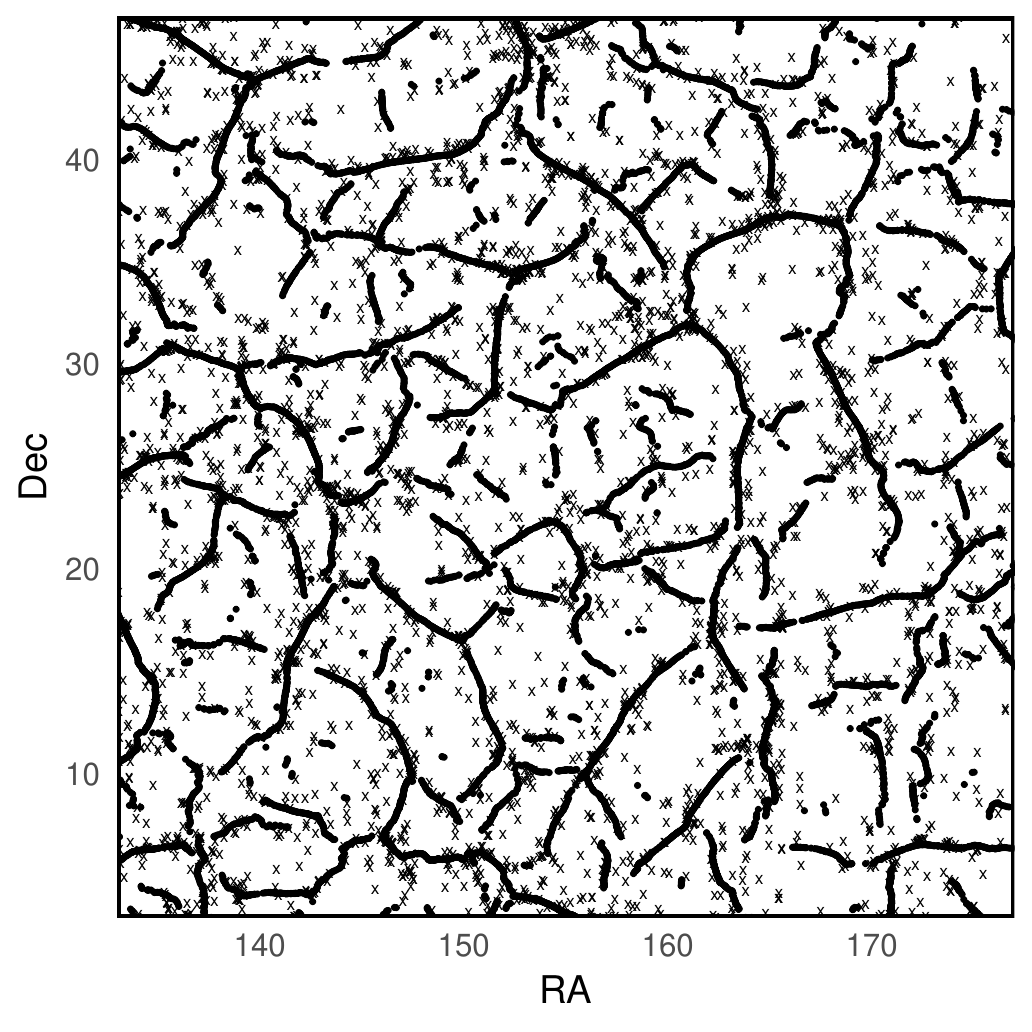}
    \label{fig:cosmic-web-LCRS-Fig13-150}
  \end{subfigure}
  \begin{subfigure}[b]{0.24\textwidth}
    \caption{$h = 2.03$}
    \includegraphics[width=\textwidth]{Figures/CosmicWeb_sLCRS-Fig13-203.pdf}
    \label{fig:cosmic-web-SCMS-Fig13-203}
  \end{subfigure}
  \begin{subfigure}[b]{0.24\textwidth}
    \caption{$h = 2.5$}
    \includegraphics[width=\textwidth]{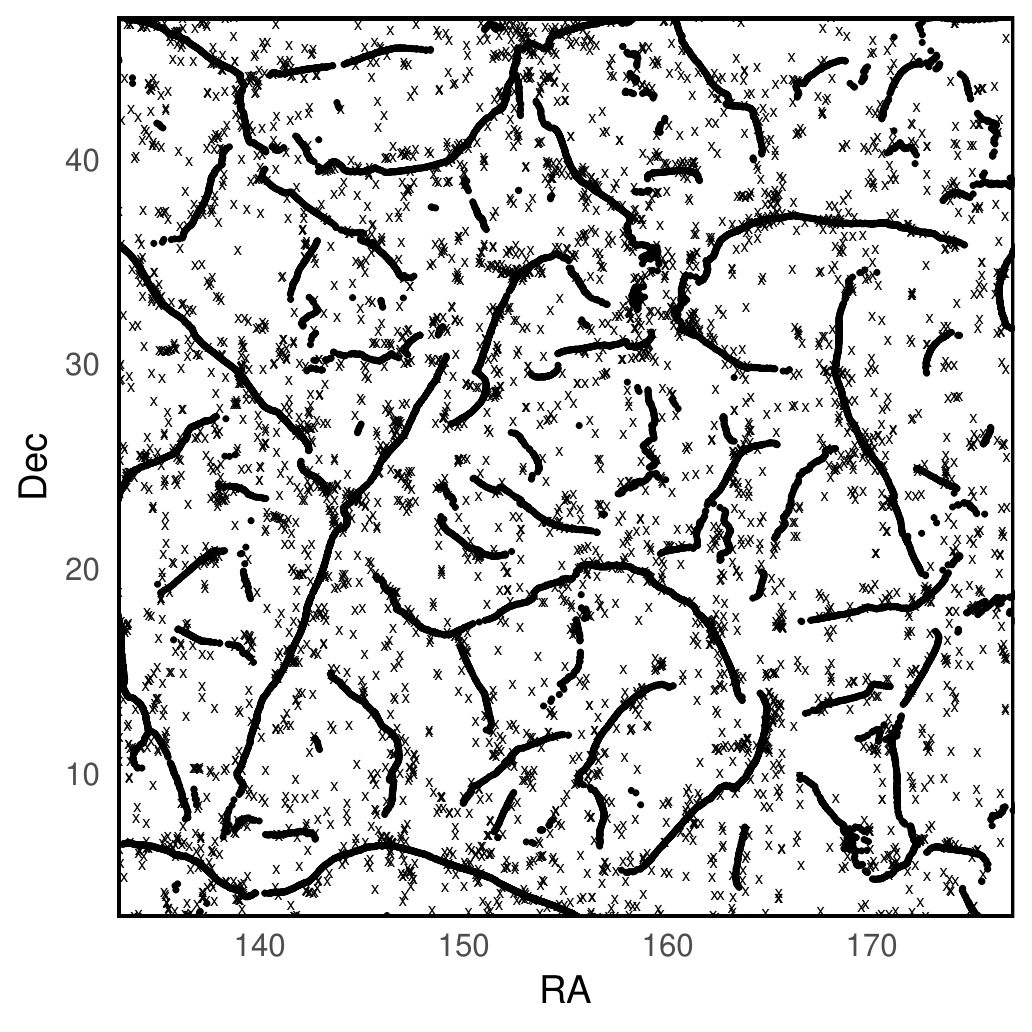}
    \label{fig:cosmic-web-LCRS-Fig13-250}
  \end{subfigure}
  \caption{Shows the estimated ridge points for the sLCRS with different
    bandwidths $h$.}
  \label{fig:sLCRS-bandwidths}
\end{figure}

\paragraph{Results.}
%

Looking at the results for the lower redshift data in Figure
\ref{fig:cosmic-web}
(\subref{fig:cosmic-web-SCMS-Fig12}--\subref{fig:cosmic-web-sLCRS-Fig12}), we
see that the ridge lines found by the SCMS are also covered by the LCRS and
sLCRS. In areas where there is a clear ridge line observable by just
looking at the data (i.e.~from $(145,20)$ to $(163,5)$), the LCRS and sLCRS
algorithms follow it smoothly and go directly through the data. The SCMS is
also smooth, but the ridge is closer towards the center of the curvature,
and lays above the most data points for smaller Dec and below for larger
Dec. Hence one observes the same phenomena as for the simulated circle data
in Figure \ref{fig:SCMS}. In areas where there isn't a clear ridge line
visible from eye (i.e~around $(145,22)$), the LCRS may be fragmented,
looking at the threshold intervals reveal, that we have a flat projected
weighted density along the smallest variance (compare Figure
\ref{fig:LCRS}), whence the ridge may lay somewhere on this flat part,
represented by the threshold interval. The sLCRS finds a smooth ridge in
this area through the data points. The ridge estimated by the SCMS is
interrupted and again biased towards the center of the curvature.

The results for the high redshift data in Figure \ref{fig:cosmic-web}
(\subref{fig:cosmic-web-SCMS-Fig13}--\subref{fig:cosmic-web-sLCRS-Fig13})
show the same effects as for the low redshift data.

In Figure \ref{fig:SCMS-bandwidths}--\ref{fig:sLCRS-bandwidths} we look at
the estimated ridges for different bandwidths for the high redshift
data. For all algorithms the number of ridge lines get lower and the ridge
lines getting longer as the bandwidth increases. This effect is most
distinctive for the SCMS algorithm. For the LCRS algorithm, even for larger
bandwidth there are still some short ridge line visible. However, in case
of the SCMS algorithm we notice, that the ridges move around for different
bandwidths, e.g.~the ridge from $(150,20)$ to $(160,5)$ goes from a
zigzag-shape to a round c-shape. In case of the LCRS and sLCRS algorithm the
estimated ridge does not move around for different bandwidths, it just gets
less connected for lower bandwidths.

\subsection{Discussion}
\label{sec:discussion-1}

In the simulated and the real data example we see that the performance of
the SCMS algorithm strongly depends on the bandwidth choice. This is
evident for the circle data with high bandwidths (Figure \ref{fig:SCMS}),
but also for the galaxy data. For the LCRS algorithm this effect
does not occur and for the sLCRS algorithm it is small. Hence, LCRS and
sLCRS is much more robust regarding bandwidth selection then the SCMS.

Therefore, whenever one is not only interested in estimating a smooth (but
possibly biased) ridge, we recommend using LCRS algorithm with threshold
intervals in case one wants a uncertainty measure or sLCRS if one wants
smooth ridges that are not biased.

\paragraph{Acknowledgements.}
This work was supported by Swiss National Science Foundation. I'm grateful
to Johanna F. Ziegel and Lutz D\"umbgen for their support and valuable inputs.



\bibliographystyle{abbrvnat}
\bibliography{../../../Literature/Literature.bib}


\appendix

  \section{Auxiliary Results}

  \subsection{Ridge of circle data}
  \label{sec:ridge-circle-data}

  Suppose we have a sample $\X_1, \X_2, \ldots, \X_n \in \R^2$ with
  distribution $\mathcal{P} := \mathcal{L}(\X)$, where
  \begin{displaymath}
    \X \defeq r
    \begin{pmatrix}
      \cos(2\pi U)\\
      \sin(2\pi U)\\
    \end{pmatrix}
    + \sigma \Z,
  \end{displaymath}
  with $r, \sigma \in \R_{>0}$ and independent random variables
  $U \sim \mathcal{U}([0,1])$ and $\Z \sim \NN(0, \I_2)$.

  Let $f$ be the probability density function of the distribution
  $\mathcal{P}$, by the law of total probability we have
  \begin{align*}
    f(\x) & = \int_0^1 f_{\Z}(\x \, \vert \, U = u) \, du\\
          & = (2 \pi \sigma^2)^{-1} \int_0^1 \exp\Bigl( -\frac{r^2 - 2r
            \bigl( x_1 \cos(2\pi u) + x_2 \sin(2\pi u) \bigr) + \Vert \x
            \Vert^2}{2 \sigma^2} \Bigr) \, du\\
          & = (4 \pi^2 \sigma^2)^{-1} \exp \Bigl( -\frac{r^2 +
            \lVert \x \rVert^2}{2\sigma^2} \Bigr) \int_0^{2\pi} \exp\Bigl(
            \frac{r}{\sigma^2} \bigl( x_1\cos(u) + x_2\sin(u) \bigr) \Bigr) \,
            du\\
          & = (2\pi \sigma^2)^{-1} I_o\bigl( r / \sigma^2 \lVert \x \rVert
            \bigr) \exp\Bigl( - \frac{r^2 + \lVert \x \rVert^2}{2 \sigma^2}
            \Bigr),
  \end{align*}
  where we used
  \begin{align*}
    \left\Vert
    \begin{pmatrix}
      r \cos(u) - x_1\\
      r \sin(u) - x_2\\
    \end{pmatrix}
    \right\Vert^2 & = r^2\bigl( \cos^2(u) + \sin^2(u) \bigr) - 2r \bigl( x_1
                    \cos(u) + x_2 \sin(u)\bigr) + x_1^2 + x_2^2\\
                  & = r^2 - 2r \bigl( x_1 \cos(u) + x_2 \sin(u) \bigr) + \lVert \x \rVert^2
  \end{align*}
  and
  \begin{displaymath}
    \int_0^{2\pi} \exp\Bigl( t \bigl( x_1 \cos(u) + x_2 \sin(u) \bigr)
    \Bigr)\, du = 2\pi I_0(t \lVert \x \rVert),
  \end{displaymath}
  with $I_0$ being the modified Bessel function of the first kind with
  parameter $0$. In general, the modified Bessel function of the first kind
  with parameter $\nu \in \mathbb{Z}$ can be written as
  \begin{displaymath}
    I_{\nu}(t) = (t / 2)^\nu \sum_{m = 0}^\infty \frac{(t^2 / 4)^m}{m!
      (m + \nu)!}.
  \end{displaymath}
  The derivative of $I_0$ has the following representation:
  \begin{displaymath}
    I_o'(t) = t / 2 \sum_{m=1}^\infty \frac{m (t^2 / 4)^{m-1}}{(m!)^2} =
    t  / 2 \sum_{m=0}^\infty \frac{(t^2 / 4)^m}{m! (m + 1)!} = I_1(t).
  \end{displaymath}

  \begin{lem}[Theorem 1 in \citet{SimpsonSpector1984}]
    \label{lem:Bessel}
    Let be $\nu(t) \defeq t I_0(t) / I_1(t)$ for $t > 0$, then $\nu(t)$ is
    strictly increasing and convex with
    \begin{displaymath}
      \nu(0+) = 2, \quad \nu'(0+) = 0 \quad \text{and} \quad \nu''(0+) = \frac{1}{3}.
    \end{displaymath}
  \end{lem}

  \begin{lem}
    \label{lem:symmetric}
    For a rotationally symmetric distribution with probability density
    function $f: \R^d \rr \R$; $f(\x) = g(\lVert \x \rVert)$ for some
    function $g: \R_{\geq 0} \rr \R_{\geq 0}$, $g'(0+) = 0$ we have
    \begin{displaymath}
      \mathrm{R}_{d-1}(f) = \bigl\{ \x \in \R^d: g'(\lVert \x \rVert) = 0,
      g''(\lVert \x \rVert) < 0 \bigr\}.
    \end{displaymath}
  \end{lem}

\begin{proof}
  It is $\bs{0} \in \mathrm{Ridge}(f)$ if, and only if, $g''(0) <
  0$. These, we will assume that $\x \neq \bs{0}$ in the remainder of this
  proof. It is
  \begin{align*}
    Df(\x) & = g'(\lVert \x \rVert) \frac{\x}{\lVert \x
             \rVert},\\
    D^2f(\x) & = g''(\lVert \x \rVert) \Bigl( \frac{\x}{\lVert \x \rVert}
               \Bigr) \Bigl( \frac{\x}{\lVert \x \rVert} \Bigr)^\top +
               \frac{g'(\lVert \x \rVert)}{\lVert \x \rVert} \Bigl( \I_2 -
               \Bigl( \frac{\x}{\lVert \x \rVert} \Bigr) \Bigl(
               \frac{\x}{\lVert \x \rVert} \Bigr)^\top\Bigr),
  \end{align*}
  because
  \begin{displaymath}
    \frac{\partial}{\partial x_i} \lVert \x \rVert = \frac{1}{2 \sqrt{x_1^2
        + \cdots +  x_2^2}} 2 x_i = \frac{x_i}{\lVert \x \rVert}.
  \end{displaymath}
  Let $\v_1 \defeq \x / \lVert \x \rVert$, for all vectors $\v \in \R^d$
  with $\v_1^\top \v = 0$ it holds
  \begin{displaymath}
    D^2f(\x) \v = \bs{0} = 0 \v.
  \end{displaymath}
  Hence the space
  $\{\v \in \R^d: \v_1^\top \v = 0\} \subset \mathrm{kern}\bigl( D^2f(\x)
  \bigr)$ and
  \begin{displaymath}
    D^2f(\x) \v_1 = g''(\lVert \x \rVert) \v_1.
  \end{displaymath}
  So, the eigenvalues of $D^2f(\x)$ are $g''(\lVert \x \rVert)$ with
  multiplicity $1$ and $0$ with multiplicity $d-1$ and
  $\x \in \mathrm{Ridge}_{d-1}(f)$ if, and only if,
  \begin{displaymath}
    D f(\x)^\top \v_1 = g'(\lVert \x \rVert) \v_1^\top \v_1 = g'(\lVert \x
    \rVert) = 0 \quad \text{and} \quad g''(\lVert \x \rVert) < 0.
  \end{displaymath}
\end{proof}

The distribution of the circle data is rotationally symmetric. Indeed,
$f(\x) = g(\lVert \x \rVert)$ with
\begin{align*}
  g(t) & = (2\pi \sigma^2)^{-1} I_0 (r / \sigma^2 \cdot t) \exp\Bigl(
         -\frac{r^2 + t^2}{2 \sigma^2} \Bigr) ,\\
  g'(t) & = (2\pi \sigma^2)^{-1} \bigl( r / \sigma^2 \cdot I_1(r / \sigma^2
          \cdot t) - t / \sigma^2 \cdot I_0(r / \sigma^2 \cdot t) \bigr)
          \exp\Bigl( - \frac{ r^2 + t^2}{2 \sigma^2} \Bigr).
\end{align*}
Denote $\alpha \defeq r / \sigma^2 > 0$, $t = \lVert \x \rVert$. The
following statements are equivalent:
\begin{align*}
  g'(t) & = 0 &\\
  \alpha I_1(\alpha t) - t / \sigma^2 \cdot I_0(\alpha t) & = 0 & \\
  \frac{t}{\sigma^2 \alpha} \frac{I_0(\alpha t)}{I_1(\alpha t)} &= 1
              & \text{or} \quad t = 0\\
  \nu(\alpha t) = \alpha t \frac{I_o(\alpha t)}{I_1(\alpha t)}
        & = \sigma^2 \alpha ^2 = \frac{r^2}{\sigma^2} & \text{or} \quad t
                                                        = 0.\\
\end{align*}
By Lemma \ref{lem:Bessel} we have $t = 0$ if $r / \sigma \leq \sqrt{2}$ and
$t > 0$ if $r / \sigma > \sqrt{2}$. The solution is unique and can be
calculated by bisection.

\subsection{Matrix Analysis}
\label{sec:matrix-analysis}

We state some notation and two results from matrix analysis which will be
used later on. References are \citet{Bhatia1997} and \citet{YuETAL2015}.


\begin{thm}[Weyl's Inequality; Theorem III.2.1 in \citet{Bhatia1997}]
  \label{thm:weyl}
  Let $\A, \B$ be symmetric $d \times d$ matrices. Then,
  \begin{align*}
    \lambda_j(\A + \B) & \geq \lambda_i(\A) + \lambda_{j-i+1}(\B) \quad
                         \text{for } 1 \leq i \leq j,\\
    \lambda_j(\A + \B) & \leq \lambda_i(\A) + \lambda_{j-i+d}(\B) \quad
                         \text{for } j \leq i \leq d.
  \end{align*}
  Consequently, for each $1 \leq j \leq d$,
  \begin{displaymath}
    \lambda_j(\A) + \lambda_d(\B) \leq \lambda_j(\A + \B) \leq \lambda_j(\A)
    + \lambda_1(\B).
  \end{displaymath}
\end{thm}

\begin{thm}[Davis-Kahan sin$\Theta$ Theorem: Theorem VII.3.4 in
  \citet{Bhatia1997}, Theorem 1 in \citet{YuETAL2015}]
  \label{thm:davis-kahan}
  Let $\A, \B$ be symmetric $d \times d$ matrices and $1 \leq s < d$ such
  that $\delta \defeq \lambda_s(\A) - \lambda_{s+1}(\B) > 0$. Then,
  \begin{multline*}
    \frac{\lVert \V_{\!\!\perp}(\A) \V_{\!\!\perp}(\A)^\top -
      \V_{\!\!\perp}(\B) \V_{\!\!\perp}(\B)^\top \rVert_F}{\sqrt{2}} =
    \lVert \sin\Theta( \V_{\!\!\perp}(\A), \V_{\!\!\perp}(\B)) \rVert_F\\ =
    \lVert \V_{\!\!\perp}(\A) \V_{\!\!\perp}(\A)^\top \Vp(\B) \Vp(\B)^\top
    \rVert_F \leq \frac{\lVert \A - \B\rVert_F}{\delta},
  \end{multline*}
  where
  $\Theta(\V_{\!\!\perp}(\A), \V_{\!\!\perp}(\B)) \in \R^{d \times d}$ is a
  diagonal matrix with the vector
  $(\cos^{-1}(\sigma_1), \ldots, \cos^{-1}(\sigma_d))^\top$, with
  $\sigma_1 \geq \cdots \geq \sigma_d \geq 0$ being the singular values of
  $\V_{\!\!\perp}(\A)^\top \V_{\!\!\perp}(\B)$, on the diagonal. The
  function $\sin\Theta(\V_{\!\!\perp}(\A), \V_{\!\!\perp}(\B))$ is defined
  entry-wise; see \citet{YuETAL2015}.
\end{thm}

The first two equalities follow from the definition of the \emph{angle
  operator} and Exercise VII.1.11 in \citet{Bhatia1997}.

\section{Proofs}
\label{sec:proofs}
In the following proofs we will suppress the argument $\x$.

\subsection{Ridges}
\label{app:ridges}

\begin{proof} [\textnormal{\textbf{Proof of Theorem \ref{thm:hessian_variance}}}]
  We write \begin{displaymath} \bs{\Sigma}_h = \tilde{\bs{\Sigma}}_h +
    \M_h,
  \end{displaymath}
  where
  \begin{displaymath}
    \tilde{\bs{\Sigma}}_h \defeq h^2 \I_d + h^4 D^2\ell =
    \V(D^2\ell)\bigl( h^2 \I_d + h^4 \bLambda(D^2\ell) \bigr)
    \V(D^2\ell)^\top
  \end{displaymath}
  and
  \begin{displaymath}
    \M_h \defeq \bs{\Sigma}_h - \tilde{\bs{\Sigma}}_h \quad \text{with}
    \quad \lVert \M_h \rVert_F = \lVert \bLambda(\M_h) \rVert_F = \o(h^4),
  \end{displaymath}
  by Lemma \ref{lem:conditional-distribution}. It is
  \begin{displaymath}
    \bLambda(\tilde{\bs{\Sigma}}_h) = h^2 \I_d + h^4 \bLambda(D^2\ell)
    \quad \text{and} \quad \V(\tilde{\bs{\Sigma}}_h) = \V(D^2\ell).
  \end{displaymath}
  By Weyl's inequality is
  \begin{displaymath}
    \lambda_s(\bs{\Sigma}_h) - \lambda_{s+1}(\tilde{\bs{\Sigma}}_h) \geq
    \lambda_s(\tilde{\bs{\Sigma}}_h) + \lambda_d(\M_h) -
    \lambda_{s+1}(\tilde{\bs{\Sigma}}_h) = h^4 \delta + \lambda_d(\M_h) \geq
    h^4 2^{-1} \delta
  \end{displaymath}
  for $h$ sufficiently small to achieve
  $\lambda_d(\M_h) \ge -h^4 2^{-1} \delta$, where we used
  \begin{displaymath}
    \delta = \lambda_s(D^2\ell) - \lambda_{s+1}(D^2\ell) =
    h^{-4}\bigl(\lambda_s(\tilde{\bs{\Sigma}}_h) -
    \lambda_{s+1}(\tilde{\bs{\Sigma}}_h)\bigr).
  \end{displaymath}
  By the Davis-Kahan $\sin\Theta$ Theorem is then
  \begin{displaymath}
    \text{dist}\bigl( \Vo(\bs{\Sigma}_h), D^2\ell \bigr) \leq \frac{\lVert 
      \bs{\Sigma}_h - \tilde{\bs{\Sigma}}_h \rVert_F}{h^4 \delta / \sqrt{2}} =
    \frac{\sqrt{2} \lVert \M_h \rVert_F}{h^4 \delta} \rightarrow 0 \quad \text{as
    } h \rightarrow 0.
  \end{displaymath}
\end{proof}

\begin{proof}[\textnormal{\textbf{Proof of Theorem \ref{thm:hessian_variance2}}}]
  Let $\Vo \defeq \Vo(D^2f(\x))$ and $\V \defeq \V(D^2f(\x))$. The gradient
  and Hessian matrix of $\z'' \mapsto \ell(\x + \Vo z'')$ are
  \begin{displaymath}
    D\bigl(\ell( \x + \Vo \z'')\bigr) = \Vo^\top D\ell(\x + \Vo \z'')
  \end{displaymath}
  and
  \begin{displaymath}
    D^2\bigl(\ell(\x + \Vo \z'')\bigr) = \Vo^\top D^2\ell(\x + \Vo \z'') \Vo,
  \end{displaymath}
  respectively, with
  \begin{displaymath}
    D\ell(\y) = f(\y)^{-1} Df(\y) \quad \text{and} \quad D^2\ell(\y) =
    f(\y)^{-1}D^2f(\y) - f(\y)^{-2} Df(\y) Df(\y)^\top.
  \end{displaymath}
  The matrix $D^2\bigl(\ell(\x + \Vo \z'')\bigr)$ is negative definite,
  whenever $\u^\top D^2f(\x+ \Vo \z'')\u < 0$ for any unit vector in the
  column space of $\Vo$. We have
  \begin{align*}
    \u^\top D^2f(\x + \Vo \z'') \u
    & = \u^\top D^2f(\x) \u +  \u^\top\bigl( 
      D^2f(\x + \Vo \z'') - D^2f(\x)\bigr) \u\\
    & \leq \u^\top \V \bLambda(D^2f(\x)) \V^\top \u + \u^\top\bigl( 
      D^2f(\x + \Vo \z'') - D^2f(\x)\bigr) \u\\
    & \leq \u^\top \Vo \diag\bigl( \lambda_{s+1}(D^2f(\x)), \ldots,
      \lambda_d(D^2f(\x)) \bigr) \Vo^\top \u\\
    & \qquad + \bigl\lvert \u^\top \bigl( D^2f(\x
      + \Vo \z'')- D^2f(\x) \bigr) \u \bigr\rvert\\
    & \leq \sum_{j = s + 1}^d \lambda_{j}(D^2f(\x)) \bigl( \u^\top
      \v_j(D^2f(\x)) \bigr)^2\\
    & \qquad+ \sup_{\v \in \mathcal{S}^{d-1}} \bigl\lvert \u^\top \bigl( D^2f(\x
      + \Vo \z'')- D^2f(\x) \bigr) \u \bigr\rvert\\
    & \leq \lambda_{s+1}(D^2f(\x)) + \bigl\lVert D^2f(\x + \Vo \z'') -
      D^2f(\x) \bigr\rVert_F,
  \end{align*}
  where we used $\Vo \Vo^\top \u = \u$ and the fact, that
  $\sup_{\v \in \mathbb{S}^{d-1}} \lvert \v^\top \A \v \rvert \leq \lVert
  \A \rVert_F$ for any symmetric matrix $\A \in \R^{d \times d}$. By
  continuity of $D^2f$, there exists $\varepsilon > 0$ such that
  $\bigl\lVert D^2f(\x + \Vo \z'') - D^2f(\x) \bigr\rVert_F < \lvert
  \lambda_{s+1}(D^2f(\x)) \rvert$ for any $\z'' \in \R^{d-s}$ with
  $\lVert \z'' \rVert < \varepsilon$ and so
  $\u^\top D^2f(\x + \Vo \z'') \u < 0$ , because
  $\lambda_{s+1}(D^2f(\x)) < 0$.

  The function $\z'' \mapsto f(\x + \Vo \z'')$ has a mode at
  $\bs{0}_{d-s}$, because
  \begin{displaymath}
    D\bigl(\ell(\x + \Vo \z'')\bigr)\Big\vert_{\z'' = \ \bs{0}} = \Vo^\top
    D\ell(\x) =   f(\x)^{-1} \Vo^\top Df(\x) = 0.
  \end{displaymath}
\end{proof}

\begin{proof}[\textnormal{\textbf{Proof of Theorem
      \ref{thm:hessian_variance3}}}]
  By Theorem \ref{thm:hessian_variance2} we know, that the function
  $t \mapsto f(\x + t\u)$ has a mode at $t = 0$ for any $\u$ in the column
  space of $\Vo(D^2\ell(\x))$. Therefore, the directional derivatives in
  directions $\v_j \defeq \v_j(D^2\ell(\x))$ for $s+1 \leq j \leq d$ are
  equal to zero:
  \begin{displaymath}
    0 = \frac{d}{dt}\Big\vert_{t = 0} \ell(\x + t \v_j(D^2\ell(\x))) =
    \v_j(D^2\ell(\x))^\top D\ell(\x),
  \end{displaymath}
  whence $\Vo(D^2\ell(\x))^\top D\ell(\x) = \bs{0}$. Furthermore, the
  second directional derivatives are strictly negative. So,
  \begin{displaymath}
    0 > \frac{d^2}{dt^2} \ell(\x + t\v_j) \Big\vert_{t = 0} = \v_j^\top
    D^2\ell(\x) \v_j = \sum_{i=1}^d \lambda_i(D^2\ell(\x)) \v_j^\top \v_i
    \v_i^\top \v_j = \lambda_j(D^2\ell(\x))
  \end{displaymath}
  for $s + 1 \leq j \leq d$
\end{proof}

\begin{proof}[\textnormal{\textbf{Proof of Theorem \ref{thm:projected-density}}}]
  We write $K = e^\psi$ for some $\psi \in \cC^2(\supp(K))$ and
  $\V \defeq \V(\Sigmahx)$, then the logarithm of the integrand of $g_h$ is
  \begin{displaymath}
    \log\bigl( K(\z) f(\x + \V \z) \bigr) = \psi(\z) + \ell(\x + \V \z)
  \end{displaymath}
  and Hessian matrix
  \begin{displaymath}
    D^2\psi(\z) + h^2 \V^\top D^2\ell(\x + \V\z)\V.
  \end{displaymath}
  Because $D^2\ell$ is bounded by assumption, there exists $h_o > 0$ such
  that the Hessian matrix is negative definite for all $0 < h < h_o$.
\end{proof}

\subsection{Algorithms}
\label{app:algorithm}
In the following, we will use local moments defined as
\begin{displaymath}
  s_n^{\balpha}(\x) \defeq \frac{1}{n} \sum_{i=1}^n K_h(h^{-1}(\X_i -
  \x)) h^{-\lvert \balpha \rvert}(\X_i - \x)^{\alpha} \quad \text{for }
  \balpha \in \N_0^d \quad \text{and} \quad s_n(\x) \defeq s_n^{\bs{0}}(\x).
\end{displaymath}
\begin{proof}[\textnormal{\textbf{Proof of Theorem \ref{thm:Lipschitz}}}]
  Let $\tilde{L}, m > 0$ and $M > 1$ be the constants as given in Lemma
  \ref{lem:lipschitz-1} and \ref{lem:lipschitz-2} stated below. Then,
  \begin{align*}
    \lVert \hat{\bs{\Sigma}}_{n,h}(\x) - \hat{\bs{\Sigma}}_{n,h}(\y) \rVert_F^2
    & = \sum_{1 \leq i,j \leq d} \left( \frac{s_n^{\e_i +
      \e_j}(\x)}{s_n(\x)} - \frac{s_n^{\e_i +
      \e_j}(\y)}{s_n(\y)} + \frac{s_n^{\e_i}(\x)
      s_n^{\e_j}(\x)}{s_n(\y)^2} - \frac{s_n^{\e_i}(\y)
      s_n^{\e_j}(\x)}{s_n(\y)^2} \right)^2\\
    & \leq 2 \sum_{1 \leq i,j \leq d} \left( \Bigl( \frac{s_n^{\e_i +
      \e_j}(\x)}{s_n(\x)} - \frac{s_n^{\e_i +
      \e_j}(\y)}{s_n(\y)}  \Bigr)^2 + \Bigl( \frac{s_n^{\e_i}(\x)
      s_n^{\e_j}(\x)}{s_n(\x)^2} - \frac{s_n^{\e_i}(\y)
      s_n^{\e_j}(\y)}{s_n(\y)^2} \Bigr)^2 \right)\\
    & \leq 2 \sum_{1 \leq i,j \leq d} \left( \Bigl( \frac{2 M\tilde{L}}{m^2}
      \Bigr)^2 \lVert \x - \y \rVert^2 + \Bigl( \frac{4 M^2 \tilde{L}}{m^2}
      \Bigr)^2 \lVert \x - \y \rVert^2 \right)\\
    & \leq 2 d^2 \Bigl( \frac{4M^2 \tilde{L}}{m^2} \Bigr)^2 \lVert \x - \y
      \rVert^2, 
  \end{align*}

  This shows,
  \begin{displaymath}
    \lVert \hat{\bs{\Sigma}}_{n,h}(\x) - \hat{\bs{\Sigma}}_{n,h}(\y) \rVert_F \leq
    \frac{4dM\sqrt{\tilde{L}}}{m^2} \lVert \x - \y \rVert.
  \end{displaymath}
\end{proof}

\begin{lem} \label{lem:lipschitz-1} For a sample
  $\mathcal{X} = \{\X_1, \ldots, \X_n\}$, fixed $h$ and a Kernel $K$ with
  $\sup_{\lvert \bgamma \rvert = 1}\lVert K^{(\bgamma)} \rVert_\infty <
  \infty$ , the local sample moments up to order 2 are Lipschitz continuous
  on the convex hull of the sample, i.e. there exists $\tilde{L} > 0$, such
  that
  \begin{displaymath}
    \lvert s_n^{\balpha}(\x) - s_n^{\balpha}(\y) \rvert \leq \tilde{L}
    \lVert \x - \y 
    \rVert \quad \text{for} \ \x, \y \in \conv(\X_1, \ldots, \X_n) \quad
    \text{for all } \balpha \in \N_0^d \text{ with } \lvert \balpha \rvert
    \leq 2
  \end{displaymath}
\end{lem}

\begin{proof}
  It is
  \begin{align}
    s_n^{\balpha}(\y)
    & = \frac{1}{n} \sum_{i=1}^n K_h(\X_i - \y)(\X_i - \x + \x -
      \y)^\balpha  \nonumber\\ 
    & = \frac{1}{n} \sum_{i=1}^n K_h(\X_i - \y) \sum_{\bgamma \leq \balpha}
      \binom{\balpha}{\bgamma} (\x - \y)^\bgamma (\X_i - \x)^{\balpha -
      \bgamma}\nonumber\\
    & = s_n^\balpha(\x)  + \frac{1}{n} \sum_{i=1}^n K_h(\X_i - \y)
      \sum_{\substack{\lvert  \bgamma \rvert \geq 1\\ \bgamma \leq \balpha}}
    \binom{\balpha}{\bgamma} (\x - \y)^\bgamma (\X_i - \x)^{\balpha -
    \bgamma} \label{eq:lipschitz-moment-1}\\ 
    & \quad + \frac{1}{n} \sum_{i=1}^n (\X_i - \x)^\balpha \bigl( K_h(\X_i
      - \y) - K_h(\X_i - \x) \bigr), \label{eq:lipschitz-moment-2}
  \end{align}
  where
  \begin{displaymath}
    \binom{\balpha}{\bgamma} \defeq \frac{\balpha!}{\bgamma!(\balpha - \bgamma)!}
  \end{displaymath}
  and $\bgamma \leq \balpha$ is to be understood component-wise.  The
  absolute value of \eqref{eq:lipschitz-moment-1} can then be bounded by
  \begin{align*}
    \lVert K \rVert_\infty \lVert \x - \y \rVert\sum_{\substack{\lvert \bgamma
    \rvert \geq 1\\ 
    \bgamma \leq \balpha}} \binom{\balpha}{\bgamma} \diam(\mathcal{X})^{\lvert
    \bgamma \rvert - 1} \leq 3 \lVert K \rVert_\infty \max(\diam(\mathcal{X}),
    1)  \lVert \x - \y \rVert,
  \end{align*}
  because
  \begin{displaymath}
    \sum_{\substack{\lvert \bgamma \rvert \geq 1\\ \bgamma \leq \balpha}}
    \binom{\balpha}{\bgamma} = 
    \begin{cases}
      1 & \text{for } \lvert \balpha \rvert = 1,\\
      3  & \text{for } \lvert \balpha \rvert = 2.\\
    \end{cases}
  \end{displaymath}

  The absolute value of \eqref{eq:lipschitz-moment-2} is bounded by
  \begin{multline*}
    \max(\diam(\mathcal{X}), 1)^{\lvert \balpha \rvert} h^{-d}\left\lvert
      K\bigl( h^{-1}(\X_i - \y) \bigr) - K\bigl( h^{-1}(\X_i - \x) \bigr)
    \right\rvert\\
    \leq \max(\diam(\mathcal{X}), 1)^{\lvert \balpha \rvert} h^{-d-1}
    \sup_{\lvert \bgamma \rvert = 1}\lVert K^{(\bgamma)} \rVert_\infty
    \lVert \x - \y \rVert.
  \end{multline*}

  Hence,
  \begin{displaymath}
    \lvert s_n^\balpha(\y) - s_n^\balpha(\x) \rvert \leq 
    \max(\diam(\mathcal{X}), 1)^2 \sup_{\lvert \bgamma \rvert = 1}\lVert
    K^{(\bgamma)} \rVert_\infty  (3 + h^{-2})
    \lVert \y - \x \rVert,
  \end{displaymath}
  for any $\balpha$ with $0 \leq \lvert \balpha \rvert \le 2$.
\end{proof}

\begin{lem} \label{lem:lipschitz-2} The quotient
  $s_n^\balpha(\x) / s_n(\x)$ is Lipschitz continuous, whenever there
  exists $0 < \tau \le s_n(\x)$ for $\x \in \conv(\X_1, \ldots, \X_n)$.
\end{lem}
\begin{proof}
  Let
  \begin{displaymath}
    M \defeq \max_{\substack{0 \leq \lvert \balpha \rvert \leq 2\\ \z \in
        \conv(\mathcal{X})}} s_n^\balpha(\z),
  \end{displaymath}
  then
  \begin{align*}
    \left\vert \frac{s_n^\balpha(\y)}{s_n(\y)} -
    \frac{s_n^\balpha(\x)}{s_n(\x)} \right\vert
    & \leq \frac{1}{m^2} \lvert
      s_n^\balpha(\y) s_n(\x) - s_n^\balpha(\x) s_n(\y)
      \rvert\\
    & \leq \frac{1}{m^2} \bigl( \lvert s_n^\balpha(\y) s_n(\x) -
      s_n^\balpha(\x) s_n(\x) \rvert + \lvert s_n^\balpha(\x)
      s_n(\x) - s_n^\balpha(\x) s_n(\y) \rvert \bigr)\\
    & \leq \frac{1}{m^2}\bigl( \lvert s_n(\x) \rvert \cdot \lvert
      s_n^\balpha(\y) - s_n^\balpha(\x) \rvert + \lvert s_n^\balpha(\x) \rvert
      \cdot \lvert s_n(\x) - s_n(\y) \rvert\bigr)\\
    & \leq \frac{2 M \tilde{L}}{m^2} \lVert \x - \y \rVert,
  \end{align*}
  for any $\x, \y \in \conv(\mathcal{X})$.
\end{proof}

\paragraph{Smoothed ridge.}

The following calculations are useful to find the unique mode of $\hat{g}*$
via the Newton method. We have
\begin{displaymath}
  \hat{g}^*(y) = \sum_{j=2}^m \hat{g}_{j-1} \int_{x_{j-1}}^{x_j}
  \exp\bigl( \hat{s}_j (t - x_{j-1}) \bigr) \phi_\gamma(y - t) \, dt =
  \sum_{j = 2}^m \hat{f}_{j-1} q_\gamma(y, \hat{s}_j, x_{j-1}, x_j),
\end{displaymath}
where $x_1 < x_2 < \ldots < x_m$ denote the knots of the log-density
estimator $\hat{g}$, $\hat{g}_{j-1} \defeq \hat{g}(x_{j-1})$ and
$\hat{s}_j \defeq \frac{\log(\hat{f}_j) - \log(\hat{f}_{j-1})}{x_j -
  x_{j-1}}$ for $2 \leq j \leq m$. The auxiliary function $q_\gamma$ is
\begin{displaymath}
  q_\gamma(x, a, u, v) \defeq \int_u^v e^{a(x - u)} \phi_\gamma(x - y) \,
  dt = e^{a(x-u) + a^2 \gamma^2 / 2} \Bigl( \Phi\bigl( \frac{v-x-a
    \gamma^2}{\gamma} \bigr) - \Phi\bigl( \frac{u - x - a \gamma^2}{\gamma}
  \bigr) \Bigr),
\end{displaymath}
\begin{align*}
  q_\gamma'(x, a, u, v)
  & = a q_\gamma(x, a, u, v) + e^{a(x-u) + a^2 \gamma^2
    / 2} \gamma^{-1} \Bigl( \phi\bigl( \frac{u-x- a \gamma^2}{\gamma}\bigr) -
    \phi \bigl( \frac{v - x - a \gamma^2}{\gamma} \bigr)\Bigr)\\
  & = a q_\gamma(x, a, u, v) + \frac{\gamma^{-1}}{\sqrt{2\pi}} \Bigl(
    e^{a(x-u) + a^2 \gamma^2 / 2 - \ell^2 / 2} - e^{a(x-u) + a^2 \gamma^2 /
    2 - r^2 / 2} \Bigr)
\end{align*}
and
\begin{align*}
  q_\gamma''(x, a, u, v)
  & = 2 a q_\gamma'(x, a, u, v) - a^2 q_\gamma(x, a,
    u, v) + e^{a(x-u) + a^2 \gamma^2 / 2} \gamma^{-2} \Bigl( \phi'\bigl(
    \frac{v-x- a \gamma^2}{\gamma} \bigr) - \phi'\bigl( \frac{u - x - a
    \gamma^2}{\gamma} \bigr) \Bigr)\\
  & = 2 a q_\gamma'(x, a, u, v) - a^2 q_\gamma(x, a, u, v) \frac{\gamma^{-2}}{\sqrt{2\pi}}\Bigl( \ell e^{a(x-u) + a^2
    \gamma^2 / 2 - \ell^2 / 2} - r e^{a(x-u) + a^2 \gamma^2 / 2 - r^2 / 2} \Bigr),
\end{align*}
where $\phi$ and $\Phi$ is the density and distribution function of a
standard normal, respectively, and
\begin{displaymath}
  r \defeq \frac{v - x - a \gamma^2}{\gamma} \quad \text{and} \quad \ell
  \defeq \frac{u - x - a \gamma^2}{\gamma}.
\end{displaymath}
Note that
\begin{displaymath}
  \phi'\bigl(
  \frac{v-x- a \gamma^2}{\gamma} \bigr) - \phi'\bigl( \frac{u - x - a
    \gamma^2}{\gamma} \bigr) = \frac{u - x - a \gamma^2}{\gamma} \phi\bigl(
  \frac{u - x - a \gamma^2}{\gamma} \bigr) - \frac{v - x - a
    \gamma^2}{\gamma} \phi\bigl( \frac{v - x - a \gamma^2}{\gamma} \bigr).
\end{displaymath}

Thus,
\begin{align*}
  \hat{f}^*(x) & = \sum_{j = 2}^m \hat{f}_{j-1} q_\gamma(x, \hat{s}_j,
                 x_{j-1}, x_j),\\
  \hat{f}^*(x)' & = \sum_{j = 2}^m \hat{f}_{j-1} q_\gamma'(x, \hat{s}_j,
                  x_{j-1}, x_j),\\
  \hat{f}^*(x)'' & = \sum_{j = 2}^m \hat{f}_{j-1} q_\gamma''(x, \hat{s}_j,
                   x_{j-1}, x_j).\\
\end{align*}

\end{document}